\documentclass[numbook, envcountsect, envcountsame, envcountreset, runningheads, smallextended]{svjour3}
\smartqed

\usepackage{amsfonts}
\usepackage{amssymb}
\usepackage{amsmath}
\usepackage{enumitem}

\usepackage{verbatim}

\usepackage{mathptmx}

\newenvironment{proofenum}{\begin{enumerate}[leftmargin=0pt,itemindent=2.5\parindent,labelsep=0pt,labelwidth=2.5\parindent,align=left]}{\end{enumerate}}

\numberwithin{equation}{section}
\numberwithin{theorem}{section}

\newcommand{\N}{\mathbb{N}}
\renewcommand{\P}{\mathbb{P}}
\newcommand{\Q}{\mathbb{Q}}
\newcommand{\R}{\mathbb{R}}
\renewcommand{\S}{\mathbb{S}}
\newcommand{\W}{\mathcal{W}}
\newcommand{\olW}{\W}

\newcommand{\lrparen}[1]{\left(#1\right)}
\newcommand{\lparen}[1]{\left(#1\right.}
\newcommand{\rparen}[1]{\left.#1\right)}
\newcommand{\lrsquare}[1]{\left[#1\right]}
\newcommand{\lsquare}[1]{\left[#1\right.}
\newcommand{\rsquare}[1]{\left.#1\right]}
\newcommand{\lrcurly}[1]{\left\{#1\right\}}
\newcommand{\lcurly}[1]{\left\{#1\right.}
\newcommand{\rcurly}[1]{\left.#1\right\}}

\newcommand{\inlrparen}[1]{(#1)}
\newcommand{\inlrsquare}[1]{[#1]}
\newcommand{\inlrcurly}[1]{\{#1\}}

\newcommand{\Ft}[1]{\mathcal{F}_{#1}}

\newcommand{\Ldpshort}[2]{L^{#1}_{#2}}

\newcommand{\LdpF}[1]{\Ldpshort{p}{#1}}
\newcommand{\LdqF}[1]{\Ldpshort{q}{#1}}
\newcommand{\LdiF}[1]{\Ldpshort{\infty}{#1}}
\newcommand{\LdoF}[1]{\Ldpshort{1}{#1}}

\newcommand{\LdpK}[3]{\Ldpshort{#1}{#2}(#3)}

\newcommand{\Lp}[2]{L^{#1}({#2})}
\newcommand{\Lpshort}[2]{L^{#1}_{#2}(\R)}
\newcommand{\LpF}[1]{\Lpshort{p}{#1}}

\newcommand{\LiF}[1]{\Lpshort{\infty}{#1}}

\newcommand{\E}[1]{\mathbb{E}\lrsquare{#1}}
\newcommand{\EQ}[1]{\mathbb{E}^{\mathbb{Q}}\lrsquare{#1}}

\newcommand{\ER}[1]{\mathbb{E}^{\mathbb{R}}\lrsquare{#1}}
\newcommand{\ES}[1]{\mathbb{E}^{\mathbb{S}}\lrsquare{#1}}
\newcommand{\Et}[2]{\E{\left.#1 \right| \mathcal{F}_{#2}}}
\newcommand{\EQt}[2]{\EQ{\left.#1 \right| \mathcal{F}_{#2}}}

\newcommand{\ERt}[2]{\ER{\left.#1 \right| \mathcal{F}_{#2}}}
\newcommand{\ESt}[2]{\ES{\left.#1 \right| \mathcal{F}_{#2}}}

\newcommand{\inE}[1]{\mathbb{E}\inlrsquare{#1}}
\newcommand{\inEQ}[1]{\mathbb{E}^{\mathbb{Q}}\inlrsquare{#1}}

\newcommand{\inES}[1]{\mathbb{E}^{\mathbb{S}}\inlrsquare{#1}}
\newcommand{\inEt}[2]{\inE{#1 \mid \mathcal{F}_{#2}}}
\newcommand{\inEQt}[2]{\inEQ{#1 \mid \mathcal{F}_{#2}}}

\newcommand{\inESt}[2]{\inES{#1 \mid \mathcal{F}_{#2}}}

\newcommand{\dQdP}{\frac{d\mathbb{Q}}{d\mathbb{P}}}
\newcommand{\dQndP}[1]{\frac{d\mathbb{Q}_{#1}}{d\mathbb{P}}}
\newcommand{\dQidP}{\dQndP{i}}
\newcommand{\dRdP}{\frac{d\mathbb{R}}{d\mathbb{P}}}

\newcommand{\dSdP}{\frac{d\mathbb{S}}{d\mathbb{P}}}
\newcommand{\dSidP}{\frac{d\mathbb{S}_i}{d\mathbb{P}}}

\newcommand{\genseq}[4]{(#1_#4)_{#4=#2}^{#3}}
\newcommand{\seq}[1]{\genseq{#1}{0}{T}{t}}

\newcommand{\FYvblank}{\tilde{F}_{(Y,v)}^{t}}
\newcommand{\FYv}[1]{\FYvblank[#1]}
\newcommand{\FYvz}[1]{\tilde{F}_{(Y,v)}^{0}[#1]}
\newcommand{\FQwblank}{F_{(\mathbb{Q},w)}^{t}}
\newcommand{\FQw}[1]{F_{(\mathbb{Q},w)}^{t}[#1]}
\newcommand{\FQwz}[1]{F_{(\mathbb{Q},w)}^{0}[#1]}

\newcommand{\trans}[1]{#1^{\mathsf{T}}}
\newcommand{\transp}[1]{\trans{(#1)}}
\newcommand{\prp}[1]{#1^{\perp}}

\newcommand{\plus}[1]{#1^+}
\newcommand{\plusp}[1]{\plus{(#1)}}

\newcommand{\diag}[1]{#1}
\newcommand{\cl}{\operatorname{cl}}
\newcommand{\co}{\operatorname{co}}

\newcommand{\as}{\text{a.s.}}

\newcommand{\Pas}{\;\P\text{-}\as}

\newcommand{\1}{\mathbf{1}}

\DeclareMathOperator*{\esssup}{ess\,sup}
\DeclareMathOperator*{\essinf}{ess\,inf}

\title{Multiportfolio time consistency for set-valued convex and coherent risk measures}
\author{Zachary Feinstein \and Birgit Rudloff}
\journalname{Finance and Stochastics}
\date{Received: date / Accepted: date}
\institute{Z. Feinstein \at
              Department of Electrical and Systems Engineering, Washington University in St. Louis, St. Louis, MO 63130\\
              Research supported by NSF RTG grant 0739195\\
              \email{zfeinstein@ese.wustl.edu}
           \and
           B. Rudloff \at
              Department of Operations Research and Financial Engineering; Bendheim Center for Finance, \\Princeton University,
              Princeton, NJ 08544\\
              Research supported by NSF award DMS-1007938\\
              \email{brudloff@princeton.edu}
}

\begin{document}
\maketitle
\abstract{
Equivalent characterizations of multiportfolio time consistency are deduced for closed convex and coherent set-valued risk measures on $L^p(\Omega,\Ft{},\P;\R^d)$
with image space in the power set of $L^p(\Omega,\Ft{t},\P;\R^d)$.  In the convex case, multiportfolio time consistency is equivalent to a cocycle condition on the sum
of minimal penalty functions.  In the coherent case, multiportfolio time consistency is equivalent to a generalized version of stability of the dual variables.
As examples, the set-valued entropic risk measure with constant risk aversion coefficient is shown to satisfy the cocycle condition for its minimal penalty functions,
the set of superhedging portfolios in markets with proportional transaction costs is shown to have the stability property and in markets with convex transaction costs
is shown to satisfy the composed cocycle condition, and a multiportfolio time consistent version of the set-valued average value at risk, the composed AV@R, is given
and its dual representation deduced.
~\\
\noindent The final publication is available at Springer via http://dx.doi.org/10.1007/s00780-014-0247-6.}

\section{Introduction}
The use of risk measures to calculate capital requirements has been widely studied, beginning with the seminal work on coherent risk measures by Artzner et al.~\cite{AD97,AD99}.  In \cite{FS02,FG02} the axioms of coherency have been relaxed to define convex risk measures.

Dynamic risk measures arise in a multi-period setting where risk is defined conditionally on information known at time $t$ described by a
filtration $\seq{\mathcal{F}}$.  Time consistency is a useful property for dynamic risk measures; it gives a relation between risks at
different times.  Conceptually, a risk measure is time consistent if, a priori, it is known that at a future time one portfolio is more risky
than another then at any prior time the same relation holds as well.  For dynamic risk measures and time consistency in the scalar setting, we
refer to~\cite{AD07,R04,DS05,CDK06,RS05,BN04,FP06,CS09,CK10,AP10,FS04} for the discrete time case and \cite{FG04,D06,DPRG10}  for the
continuous time case.  In particular, an equivalent property for time consistency in the coherent case is given by the stability of the dual
probability measures as seen in~\cite{AP10,CDK06,FS04,FP06,AD07}.  In the convex case a property on the sum of penalty functions was deduced
in~\cite{FP06,CDK06,BN08,BN09,AP10}.  This property is referred to as the cocycle property in~\cite{BN08,BN09}.

When multivariate random variables, or markets with transaction costs, are considered it becomes natural to work with set-valued risk measures; in this way capital requirements can be made in a basket of currencies or assets rather than a chosen num\'{e}raire.  In the static single period framework set-valued risk measures have been studied in~\cite{JMT04,HR08,HH10,HHR10}.  The dynamic version of set-valued risk measures were studied in~\cite{FR12,TL12,FR13-survey}.  We will take our setting from~\cite{FR12}.  In that paper, set-valued dynamic risk measures were discussed, and a set-valued version of time consistency, called multiportfolio time consistency, was defined.  In the scalar framework, time consistency is equivalent to the recursive form $\rho_t(X) = \rho_t(-\rho_s(X))$ ($0 \leq t < s \leq T$), and in the set-valued framework it was proven that multiportfolio time consistency is equivalent to the set-valued recursive form $R_t(X) = R_t(-R_s(X)) := \bigcup_{Z \in R_s(X)} R_t(-Z)$.
Up to this point multiportfolio time consistency was considered only for general risk measures, but not specifically for the convex or coherent cases.  The main results for this paper are to describe equivalent properties of multiportfolio time consistency for (conditionally) convex and coherent risk measures.

In section~\ref{sec_prelim}, we will review the basic results of \cite{FR12} that are needed for the present paper. We deduce several generalizations of the dual representation results of \cite{FR12}.
Section~\ref{sec_convex} deduces an equivalent characterization of multiportfolio time consistency for set-valued normalized closed (conditionally) convex
risk measures.  This is given by a property on the sum of minimal penalty functions, called the cocycle property in the scalar case in
\cite{BN08,BN09}, and is the extension of the scalar result of~\cite{FP06,CDK06,BN08,BN09,AP10}. The proof of this result is entirely different from the proof in the scalar case as the scalar method leads to difficulties in the set-valued case, due to the union in the set-valued recursive form. As an example we consider the set-valued entropic risk measure.
Section~\ref{sec_coherent} discusses two equivalent characterizations of multiportfolio time consistency for set-valued normalized
closed (conditionally) coherent risk measures.  The first is the result for convex risk measures applied to the coherent case. This characterization
has not been explicitly stated in the scalar case, but is useful for generating multiportfolio time consistent risk measures (see
e.g.~\cite{CK10}).
The second property is the set-valued generalization of stability of the dual variables, and generalizes the work in~\cite{AP10,CDK06,FS04,FP06,AD07}. The set of superhedging portfolios in markets with proportional transaction costs will serve as an example.
Section~\ref{sec_composition} gives a method for composing a risk measure backwards in time to create a multiportfolio time consistent
version of this risk measure.  Special attention is given to the composed form for (conditionally) convex and coherent risk measures. As examples, we will study superhedging under convex transaction costs in the convex case, and the composed AV@R in the coherent case.

\section{Set-valued dynamic risk measures}
\label{sec_prelim}
In this section, we will introduce some notations and, for easing the readability of the present paper, review basic definitions and main results about duality and multiportfolio time consistency of set-valued dynamic risk measures from \cite{FR12}.

Consider a filtered probability space $\inlrparen{\Omega,\Ft{},\seq{\mathcal{F}},\P}$ satisfying the usual conditions
where $\Ft{T} = \Ft{}$.  One can consider either discrete time $\{0,1,...,T\}$ or continuous time $[0,T]$.  Let $d\geq 1$ be the number of assets under consideration. Let $|\cdot|$ denote an arbitrary norm in $\R^d$ and let $\LdpF{t} :=
\Lp{p}{\Omega,\Ft{t},\P;\R^d}$ for any $p \in \inlrsquare{1,\infty}$ (with $\LdpF{} := \LdpF{T}$).  $\LdpF{t}$ denotes the linear space of
the equivalence classes of $\Ft{t}$-measurable functions $X: \Omega \to \R^d$ such that $\|X\|_p^p =
\inlrparen{\int_{\Omega}|X(\omega)|^p d\P} < +\infty$ for any $p \in [1,\infty)$, and $\|X\|_{\infty} = \esssup_{\omega \in \Omega} |X(\omega)| < +\infty$ for $p = \infty$.
We consider the dual pair $\inlrparen{\LdpF{}, \LdqF{}}$ with $p \in \inlrsquare{1,+\infty}$ and $q$ is such that
$\frac{1}{p}+\frac{1}{q} = 1$, and endow it with the norm topology, respectively the $\sigma \inlrparen{\LdiF{},\LdoF{}}$-topology on $\LdiF{}$ in the case $p = +\infty$.

We denote by $\LdpK{p}{t}{D_t} := \inlrcurly{Z \in \LdpF{t}: \; Z \in D_t \Pas}$ those random vectors in
$\LdpF{t}$ that take $\P$-a.s. values in $D_t$.  To distinguish the spaces of random vectors from those of random variables, we will write $\LpF{t} := \Lp{p}{\Omega,\Ft{t},\P;\R}$ for the linear space of the equivalence classes of $p$ integrable $\Ft{t}$-measurable random variables $X: \Omega \to \R$. Note that an element $X \in \LdpF{t}$ has components $X_1,...,X_d$ in
$\LpF{t}$. (In-)equalities between random vectors are always understood componentwise in the
$\P$-a.s. sense.  The multiplication between a random variable $\lambda\in \LiF{t}$ and a set of random vectors
$D\subseteq\LdpF{}$ is understood in the elementwise sense, i.e. $\lambda D=\{\lambda Y:Y\in D\}\subseteq\LdpF{}$ with $(\lambda
Y)(\omega)=\lambda(\omega)Y(\omega)$. The multiplication and division between (random) vectors is understood in the componentwise sense,
i.e. $\diag{x}y := \transp{x_1 y_1,...,x_d y_d}$ and $x/y := \transp{x_1/y_1,...,x_d/y_d}$ for $x,y \in \R^d$ ($x,y \in
\LdpF{t}$) and with $y_i \neq 0$ (almost surely) for every index $i \in \{1,...,d\}$ for division.

Let $\LdpF{t,+} := \inlrcurly{X \in \LdpF{t}: \; X \in \R^d_+ \Pas}$ denote the closed convex cone of
$\R^d$-valued $\Ft{t}$-measurable random vectors with $\P$-a.s. non-negative components.  Additionally let $\LdpF{t,++}
:= \inlrcurly{X \in \LdpF{t}: \; X \in \R^d_{++} \Pas}$ be the $\Ft{t}$-measurable random vectors which are
$\P$-a.s. positive.  Similarly define $\LdpF{+} := \LdpF{T,+}$ and $\LdpF{++} := \LdpF{T,++}$.
Let $X \succeq Y$ for $X,Y \in \LdpF{}$ denote $X-Y \in \LdpF{+}$.

As in \cite{K99} and discussed in \cite{S04,KS09}, the portfolios in this paper are in ``physical units'' of an asset rather than the value in a fixed num\'{e}raire via some price.  That is, for a portfolio $X \in \LdpF{t}$, the values of $X_i$ (for $1 \leq i \leq d$) are the number of units of asset $i$ in the portfolio at time $t$.

Let us assume $m$ of the $d$ assets are eligible ($1\leq m\leq d$), that is, they can be used  to compensate for the risk of a portfolio. Without loss of generality we can assume these are the first $m$ assets, then
$M = \R^m \times \{0\}^{d-m}$ denotes the subspace of eligible assets. By section~5.4
and proposition~5.5.1 in \cite{KS09}, $M_t := \LdpK{p}{t}{M}$  is a closed (weak* closed if
$p = +\infty$) linear subspace of $\LdpF{t}$.
Let us denote $M_{t,+} := M_t \cap \LdpF{t,+}$ and $M_{t,-} := -M_{t,+}$.

A conditional risk measure is a function which maps a $d$-dimensional random variable $X$ into \[\mathcal{P}\lrparen{M_t;M_{t,+}} := \lrcurly{D \subseteq M_t: D = D + M_{t,+}},\]  which is a subset of the power set $2^{M_t}$.
Conceptually, the value of a risk measure $R_t(X)$ is the collection of eligible portfolios at time $t$ which cover the risk of the portfolio $X$.
\begin{definition}
\label{defn_conditional}
A function $R_t: \LdpF{} \to \mathcal{P}\inlrparen{M_t;M_{t,+}}$ is a \textbf{\emph{normalized (conditional) risk measure}} at time $t$ if it is
\begin{enumerate}
\item $M_t$-translative: for every $m_t \in M_t: R_t\inlrparen{X + m_t} = R_t(X) - m_t$;
\item $\LdpF{+}$-monotone: $Y \succeq X$ implies $R_t(Y) \supseteq R_t(X)$;
\item finite at zero: $\emptyset \neq R_t(0) \neq M_t$;
\item normalized: for every $X \in \LdpF{t}: R_t(X) = R_t(X) + R_t(0)$.
\end{enumerate}

Additionally, a conditional risk measure at time $t$ is \textbf{\emph{(conditionally) convex}} if for all $X,Y \in \LdpF{}$,
for all $ 0 \leq \lambda \leq 1$ ($\lambda \in \LiF{t}$ such that $0 \leq \lambda \leq 1$)
\[R_t(\lambda X + (1-\lambda)Y) \supseteq \lambda R_t(X) + (1-\lambda)R_t(Y),\]
is \textbf{\emph{(conditionally) positive homogeneous}} if for all $X \in \LdpF{}$,
for all $\lambda >0$ ($\lambda \in L^\infty_t(\R_{++})$)
\[R_t(\lambda X) = \lambda R_t(X),\]
and is \textbf{\emph{(conditionally) coherent}} if it is (conditionally) convex and (conditionally) positive homogeneous.

A conditional risk measure at time $t$ is \textbf{\emph{closed}} if the graph of the risk measure, \[\operatorname{graph}R_t = \lrcurly{(X,u) \in \LdpF{} \times M_t: \; u \in R_t(X)},\] is closed in the product topology.

A conditional risk measure at time $t$ is \textbf{\emph{convex upper continuous (c.u.c.)}} if \[R_t^{-1}(D) := \lrcurly{X \in \LdpF{}: R_t(X) \cap D
\neq \emptyset}\] is closed for any closed set $D \in \mathcal{G}(M_t;M_{t,-}) := \inlrcurly{D \subseteq M_t: D = \cl\co(D +
M_{t,-})}$. It is called \textbf{\emph{conditionally c.u.c.}} if $R_t^{-1}(D) $  is closed for any conditionally convex closed set $D \in \mathcal{G}(M_t;M_{t,-})$.
\end{definition}

The properties given in definition~\ref{defn_conditional} and their interpretations are discussed in detail in~\cite{HHR10,FR12}.
$M_t$-translativity ensures that a risk measure can be interpreted as a `capital requirement' to cover risk.  Monotonicity means that if one portfolio dominates another (almost surely) then its risk should be lower.
The normalization property (with closedness) ensures that the zero portfolio compensates the risk of the zero payoff.
Additionally, convexity (and coherence) are useful properties for measuring diversification effects of portfolios.
Clearly, a conditionally convex (conditionally positive homogeneous) function is also convex (positive homogeneous).

The image space of a closed convex conditional risk measure is given by
\[\mathcal{G}(M_t;M_{t,+}) = \lrcurly{D \subseteq M_t: D = \cl \co \lrparen{D + M_{t,+}}}.\]

Note that any c.u.c. risk measure is closed.  This follows from $(X,u) \in \operatorname{graph}R_t$ if and only if $X
+ u \in R_t^{-1}(M_{t,-})$.  In the literature, upper continuity is defined analogously to c.u.c., but with respect to all closed sets $D \subseteq M_t$ rather than the subset $\mathcal{G}(M_t;M_{t,-})$, see~\cite{HS12-uc,GRTZ03,AF90,HP97} (in the latter two
references upper continuity is called upper semicontinuity, we follow the naming practiced by~\cite{HS12-uc,GRTZ03} because upper semicontinuity can also refer to a different property for set-valued functions). As we do not need the upper continuity property for all closed sets $D \subseteq M_t$, but only for $D \in \mathcal{G}(M_t;M_{t,-})$, we labeled the corresponding property convex upper continuity.

A \textbf{dynamic risk measure} $\seq{R}$ is a sequence of conditional risk measures.  It is said to have one of the properties given in definition~\ref{defn_conditional} if $R_t$ has this property for every $t \in [0,T]$.

Instead of considering risk measures directly, a portfolio manager might be interested in the set of portfolios which have an ``acceptable" level of risk, called an acceptance set.
\begin{definition}
\label{defn_acceptance}
$A_t \subseteq \LdpF{}$ is a \textbf{\emph{conditional acceptance set}} at time $t$ if it satisfies the conditions $A_t \cap M_t \not\in \{\emptyset,M_t\} $, and $Y \succeq X$ with $X \in A_t$ implies $Y \in A_t$.
A conditional acceptance set is \textbf{\emph{normalized}} if it satisfies $A_t + A_t \cap M_t = A_t$.
\end{definition}

There is a one-to-one correspondence between risk measures and acceptance sets, see remark~2 and proposition~2.11 in~\cite{FR12}. The acceptance set associated with a conditional risk measure $R_t$ is defined by \[A_t := \lrcurly{X \in \LdpF{}: 0 \in R_t(X)}\] and the risk measure associated with an acceptance set is defined by \[R_t(X) := \lrcurly{u \in M_t: X + u \in A_t}.\] An acceptance set is called convex upper continuous (c.u.c.) if the associated risk measure is c.u.c.
Further, we will define the stepped acceptance set (from time $t$ to $s>t$) as \[A_{t,s} := \lrcurly{X \in M_{s}: 0 \in R_t(X)} = A_t \cap M_{s}.\]
For a thorough discussion of stepped risk measures see section~\ref{sec_stepped}.

%%%%%%%%%%%%%%%%%%%%%%%%%%%%%%%%%%%%%%%%
\subsection{Dual representation}
\label{subsection_dual}
Let $\mathcal{M}$ denote the set of $d$-dimensional probability measures absolutely continuous with respect to $\P$, and let $\mathcal{M}^e$ denote the set of $d$-dimensional probability measures equivalent to $\P$.
We will use a $\P$-almost sure version of the $\Q$-conditional expectation of $X \in \LdpF{}$ (for $\Q := \transp{\Q_1,...,\Q_d}\in \mathcal{M}$) given by
\[\EQt{X}{t} := \Et{\diag{\xi_{t,T}(\Q)} X}{t},\] where $\xi_{t,s}(\Q) =
\transp{\bar{\xi}_{t,s}(\Q_1),...,\bar{\xi}_{t,s}(\Q_d)}$ for any times $0 \leq t \leq s \leq T$ with \[\bar{\xi}_{t,s}(\Q_i)[\omega] :=
\begin{cases}\frac{\Et{\dQidP}{s}(\omega)}{\Et{\dQidP}{t}(\omega)} & \text{on } \Et{\dQidP}{t}(\omega) > 0\\ 1 & \text{else}
\end{cases}\] for every $\omega \in \Omega$, see e.g.~\cite{CK10,FR12}.  For any probability measure $\Q_i \ll \P$ and any times $0 \leq t \leq r \leq s \leq T$, it follows that $\dQidP = \bar{\xi}_{0,T}(\Q_i)$,
$\bar{\xi}_{t,s}(\Q_i) = \bar{\xi}_{t,r}(\Q_i) \bar{\xi}_{r,s}(\Q_i)$, and $\inEt{\bar{\xi}_{t,s}(\Q_i)}{t} = 1$ almost surely.
%%%%%%%
The half-space in $\LdpF{t}$ with normal direction $w \in \LdqF{t}\backslash \{0\}$ is denoted by
\[G_t(w) := \lrcurly{u \in \LdpF{t}: 0 \leq \E{\trans{w}u}}.\]
We will define the set of dual variables to be
\begin{align*}
\olW_t &= \lrcurly{(\Q,w) \in \mathcal{M} \times \lrparen{\plus{M_{t,+}} \backslash \prp{M_t}}:
    w_t^T(\Q,w) \in \LdqF{+}, \;\Q= \P|_{\Ft{t}}}
\end{align*}
with \[w_t^{s}(\Q,w) = \diag{w} \xi_{t,s}(\Q)\]
for any $0 \leq t \leq s \leq T$.
\[C^+=\lrcurly{v \in \LdqF{t}: \forall u \in C: \E{\trans{v}u} \geq 0}\] is the positive dual cone of a cone $C\subseteq \LdpF{t}$ for any time $t$, and
\[\prp{M_t} = \lrcurly{v \in \LdqF{t}: \forall u \in M_t: \E{\trans{v}u} = 0}.\]

In the following we
extend the duality results from~\cite{FR12} by showing  that it is sufficient to only consider
the set of probability measures which are equal to the physical probability measure $\P$ up to
time $t$, that is, to use the smaller set $\olW_t$ in the dual representation as opposed to
$\inlrcurly{(\Q,w) \in \mathcal{M} \times \inlrparen{\plus{M_{t,+}} \backslash \prp{M_t}}: w_t^T(\Q,w) \in \LdqF{+}}$
as in \cite{FR12}. It is vital for the proof of theorem~\ref{thm_mptc_stable} below to
work with this smaller set.
This result is an extension of the scalar dual representation given in~\cite{DS05,RS05,KS07}.
We will say $\Q = \P|_{\Ft{t}}$ for vector probability measures $\Q$ and some time $t\in [0,T]$ if for every $D \in\Ft{t}$ it follows that $\Q_i(D) = \P(D)$ for all $i = 1,...,d$.  In the appendix (lemma~\ref{lemma_stepped_rep} and corollary~\ref{cor_stepped_rep}) we provide the stepped version of this result.
\begin{theorem}
\label{thm_probability_equal}
A function $R_t: \LdpF{} \to \mathcal{G}(M_t;M_{t,+})$ is a \textbf{\emph{closed convex  risk measure}} if and only if
\begin{equation}
\label{convex_dual}
R_t(X) = \bigcap_{(\Q,w) \in \olW_t} \lrsquare{-\beta_t^{\min}(\Q,w) + \lrparen{\EQt{-X}{t} + G_t\lrparen{w}} \cap M_t},
\end{equation}
where $-\beta_t^{\min}$ is the minimal penalty function given by
\begin{equation}
\label{min penalty}
-\beta_t^{\min}(\Q,w) = \operatorname{cl}\bigcup_{Z \in A_t} \lrparen{\EQt{Z}{t} + G_t(w)} \cap M_t.
\end{equation}

$R_t$ is additionally coherent if and only if
\begin{equation*}
%\label{coherent_dual}
R_t(X) = \bigcap_{(\Q,w) \in \olW_{t}^{\max}} \lrparen{\EQt{-X}{t} + G_t\lrparen{w}} \cap M_t,
\end{equation*}
for
\begin{equation*}
\olW_{t}^{\max} = \lrcurly{(\Q,w) \in \olW_t: w_t^T(\Q,w) \in \plus{A_t}}.
\end{equation*}
\end{theorem}
\begin{proof}
By theorem 4.7 and corollary 4.8 in~\cite{FR12}, we have these duality results with respect to
the full dual variables
$\inlrcurly{(\Q,w) \in \mathcal{M} \times \inlrparen{\plus{M_{t,+}} \backslash \prp{M_t}}: w_t^T(\Q,w) \in \LdqF{+}}$.
It remains to show that we only need to consider the probability
measures that are equal to $\P$ on $\Ft{t}$.  We will show this for the convex case only, as
the coherent version follows similarly.

Since the dual representation only involves $w_t^T(\Q,w)$, which can be seen from
\begin{align*}
-\beta_t^{\min}&(\Q,w) + \lrparen{\EQt{-X}{t} + G_t\lrparen{w}} \cap M_t\\
&= \lrcurly{u \in M_t: \inf_{Z \in A_t} \E{\trans{w}\EQt{Z}{t}} + \E{\trans{w}\EQt{-X}{t}} \leq \E{\trans{w}u}}\\
&= \lrcurly{u \in M_t: \inf_{Z \in A_t} \E{\trans{w_t^T(\Q,w)}(Z-X)} \leq \E{\trans{w_t^T(\Q,w)}u}}
\end{align*}
for any $Z \in \LdpF{}$, we need to show
\begin{align}
\label{smallerW}
\lrcurly{w_t^T(\Q,w): (\Q,w) \in \mathcal{M} \times \lrparen{\plus{M_{t,+}} \backslash \prp{M_t}}, w_t^T(\Q,w) \in \LdqF{+}}
= \lrcurly{w_t^T(\R,v): (\R,v) \in \olW_t}
\end{align}
to prove the result.

Trivially, $\supseteq$ holds by
$\inlrcurly{(\Q,w) \in \mathcal{M} \times \inlrparen{\plus{M_{t,+}} \backslash \prp{M_t}}: w_t^T(\Q,w) \in \LdqF{+}} \supseteq \olW_t$.
For the other direction, let
$(\Q,w) \in \inlrcurly{(\Q,w) \in \mathcal{M} \times \inlrparen{\plus{M_{t,+}} \backslash \prp{M_t}}: w_t^T(\Q,w) \in \LdqF{+}}$,
and define $\R \in \mathcal{M}$ by $\dRdP = \xi_{t,T}(\Q)$.  It follows that $(\R,w) \in \olW_t$, and by construction $w_t^T(\Q,w) = \diag{w}\dRdP = w_t^T(\R,w)$.
\qed\end{proof}

The next corollary sharpens the above duality results for the conditionally convex case by using
the conditional ``half-space'' in $\LdpF{t}$ with normal direction $w \in \LdqF{t}\backslash \{0\}$
denoted by
\[\Gamma_t(w) := \lrcurly{u \in \LdpF{t}: 0 \leq \trans{w}u \; \Pas}.\]
This provides a stronger dual representation result in the spirit of decomposibility, see remark~\ref{rem_decomp} below. Note that the
set of dual variables $\olW_t$ stays the same in the conditional framework because $\plus{M_{t,+}}=\inlrcurly{v \in
\LdqF{t}: \forall u \in M_{t,+}: \trans{v}u \geq 0\; \Pas}$ and $\prp{M_t} = \inlrcurly{v \in \LdqF{t}: \forall u \in M_t: \trans{v}u = 0 \; \Pas}$.
\begin{corollary}
\label{cor_conditional_dual}
A function $R_t: \LdpF{} \to \mathcal{G}(M_t;M_{t,+})$ is a \textbf{\emph{closed conditionally convex risk measure}} if and only if
\begin{equation}
\label{conditional_convex_dual}
R_t(X) = \bigcap_{(\Q,w) \in \olW_t} \lrsquare{-\alpha_t^{\min}(\Q,w) + \lrparen{\EQt{-X}{t} + \Gamma_t\lrparen{w}} \cap M_t},
\end{equation}
where $-\alpha_t^{\min}$ is the minimal conditional penalty function given by
\begin{equation}
\label{conditional_min penalty}
-\alpha_t^{\min}(\Q,w) = \operatorname{cl}\bigcup_{Z \in A_t} \lrparen{\EQt{Z}{t} + \Gamma_t(w)} \cap M_t.
\end{equation}
$R_t$ is additionally conditionally coherent if and only if
\begin{equation}
\label{conditional_coherent_dual}
R_t(X) = \bigcap_{(\Q,w) \in \olW_{t}^{\max}} \lrparen{\EQt{-X}{t} + \Gamma_t\lrparen{w}} \cap M_t.
\end{equation}
\end{corollary}
\begin{proof}
First, we can reformulate the conditional penalty function as
\[-\alpha_t^{\min}(\Q,w) = \lrcurly{u \in M_t:
\essinf_{Z \in A_t} \trans{w}\EQt{Z}{t} \leq \trans{w}u}.\]  So, for any $(\Q,w) \in \olW_t$
\begin{align*}
-\alpha_t^{\min}(\Q,w) + \lrparen{\EQt{-X}{t} + \Gamma_t\lrparen{w}} \cap M_t= \lrcurly{u \in M_t: \essinf_{Z \in A_t} \trans{w}\EQt{Z-X}{t} \leq \trans{w}u},
\end{align*}
using $w\notin  \prp{M_t}$.
We will prove that a conditionally convex risk measure has the representation above by showing that \[\bar{R}_t(X) :=
\bigcap_{(\Q,w) \in \olW_t} \inlrsquare{-\alpha_t^{\min}(\Q,w) + \inlrparen{\inEQt{-X}{t} + \Gamma_t\inlrparen{w}} \cap M_t}
= R_t(X)\] where $R_t(X)$ is given by the dual representation \eqref{convex_dual} in theorem~\ref{thm_probability_equal}.
\begin{proofenum}
\item Let $u \in \bar{R}_t(X)$, i.e. $\trans{w}u \geq \essinf_{Z \in A_t} \trans{w}\inEQt{Z-X}{t}$ for
any dual variables $(\Q,w) \in \olW_t$.  It follows that $\inE{\trans{w}u} \geq \inE{\essinf_{Z \in A_t} \trans{w}\inEQt{Z-X}{t}}$
for any $(\Q,w) \in \olW_t$.
By the $\Ft{t}$-decomposability of the acceptance set $A_t$ (which follows from conditional convexity), i.e. $1_D A_t + 1_{D^c} A_t \subseteq A_t$ for any $D \in \Ft{t}$, we can apply \cite[theorem 1]{Y85} to
interchange the expectation and the essential infimum.  Thus, by the representation in
theorem~\ref{thm_probability_equal}, $u \in R_t(X)$.
\item Now let $u \in R_t(X)$. Assume $u \not\in \bar{R}_t(X)$, i.e. there exists some $(\Q,w) \in \olW_t$ such that
\[\P(\trans{w}u < \essinf_{Z \in A_t} \trans{w}\EQt{Z-X}{t}) > 0.\]
Define $D := \inlrcurly{\trans{w}u < \essinf_{Z \in A_t} \trans{w}\inEQt{Z-X}{t}} \in \Ft{t}$, then $1_D \trans{w}u < \essinf_{Z \in A_t} 1_D \trans{w}\inEQt{Z-X}{t}$ on $D$, and the strict
inequality also
 holds for the expectation.
As above, by conditional convexity we can interchange the expectation and
the infimum, thus we recover that
\[\E{1_D \trans{w}u} < \inf_{Z \in A_t} \E{1_D \trans{w}\EQt{Z-X}{t}}.\]
From the equality $\prp{M_t} = \inlrcurly{v \in \LdqF{t}: \forall u \in M_t: \trans{v}u = 0  \; \Pas}$, one can show that $(\Q,1_D w) \in \olW_t$, which is a contradiction to $u \in R_t(X)$.
\end{proofenum}
It remains to show that $R_t$ defined by \eqref{conditional_convex_dual} is conditionally convex.
Let $X,Y \in \LdpF{}$ and let $\lambda \in \LiF{t}$ such that $0 \leq \lambda \leq 1$.
\begin{align*}
R_t(&\lambda X + (1-\lambda) Y) = \bigcap_{(\Q,w) \in \olW_t} \lrcurly{u \in M_t: \essinf_{Z \in A_t} \trans{w}\EQt{Z-(\lambda X + (1-\lambda) Y)}{t} \leq \trans{w}u}\\
&=  \bigcap_{(\Q,w) \in \olW_t} \lrcurly{u \in M_t: \essinf_{Z_X,Z_Y \in A_t} \trans{w}\EQt{\lambda (Z_X - X) + (1-\lambda) (Z_Y - Y)}{t} \leq \trans{w}u}\\
&=  \bigcap_{(\Q,w) \in \olW_t} \lcurly{u \in M_t: \lambda \essinf_{Z \in A_t} \trans{w}\EQt{Z - X}{t} } \rcurly{+ (1-\lambda) \essinf_{Z \in A_t} \trans{w}\EQt{Z - Y}{t} \leq \trans{w}u}\\
&\supseteq \bigcap_{(\Q,w) \in \olW_t} \lsquare{\lrcurly{\lambda u: u \in M_t, \essinf_{Z \in A_t} \trans{w} \EQt{Z-X}{t} \leq \trans{w}u}}\\
    &\quad\quad+ \rsquare{ \lrcurly{(1-\lambda) u: u \in M_t, \essinf_{Z \in A_t} \trans{w}\EQt{Z-Y}{t} \leq \trans{w}u}}\\
&\supseteq \lambda \bigcap_{(\Q,w) \in \olW_t} \lrcurly{u \in M_t: \essinf_{Z \in A_t} \trans{w}\EQt{Z-X}{t} \leq \trans{w}u}\\
     &\quad\quad+ (1-\lambda) \bigcap_{(\Q,w) \in \olW_t} \lrcurly{u \in M_t: \essinf_{Z \in A_t} \trans{w}\EQt{Z-Y}{t} \leq \trans{w}u}= \lambda R_t(X) + (1-\lambda) R_t(Y).
\end{align*}
For the conditionally coherent case, we first note that $-\alpha_t^{\min}(\Q,w) \supseteq \Gamma_t(w) \cap M_t$
for every $(\Q,w) \in \olW_t$ since $0 \in A_t$ by positive homogeneity.  In fact,
\begin{align*}
-\alpha_t^{\min}(\Q,w) &= \begin{cases}G_0(w(\omega)) \cap M &\omega\in D := \lrcurly{\essinf_{Z \in A_t} \trans{w}\EQt{Z}{t} \geq 0}\\ M &\omega\in D^c\end{cases}\\
&= 1_D \Gamma_t(w) \cap M_t + 1_{D^c} M_t.
\end{align*}
Define $-\hat{\alpha}_t$ by
\[-\hat{\alpha}_t(\Q,w) = \begin{cases}\Gamma_t(w) \cap M_t &\text{if } \essinf_{Z \in A_t} \trans{w}\EQt{Z}{t} = 0 \Pas\\ M_t &\text{else}\end{cases}.\]
By construction we have $-\alpha_t^{\min} \subseteq -\hat{\alpha}_t$, and thus
\[R_t(X) \subseteq \bigcap_{(\Q,w) \in \olW_t^{\max}} \lrparen{\EQt{-X}{t} + \Gamma_t\lrparen{w}} \cap M_t.\]
Conversely, by theorem~\ref{thm_probability_equal}, $u \in R_t(X)$ if and only if $u\in M_t$ and
\begin{align*}
0 &\leq \inf_{(\Q,w) \in \olW_t^{\max}} \E{\trans{w}\EQt{X+u}{t}}.
\end{align*}
But $u \in \bigcap_{(\Q,w) \in \olW_t^{\max}} \inlrparen{\inEQt{-X}{t} + \Gamma_t\inlrparen{w}} \cap M_t$
implies $u\in M_t$ as well as $0 \leq \inf_{(\Q,w) \in \olW_t^{\max}} \inE{\trans{w}\inEQt{X+u}{t}}$, i.e. $u \in R_t(X)$.

It remains to show that $R_t$ defined by \eqref{conditional_coherent_dual} is conditionally coherent, but as this follows similarly to the
conditionally convex case above we will omit it.
\qed\end{proof}

\begin{remark}
\label{rem_decomp}
A risk measure is closed and conditionally convex if and only if it is closed, convex, and decomposable
where decomposability is defined by the equality
\[R_t(1_D X + 1_{D^c} Y) = 1_D R_t(X) + 1_{D^c} R_t(Y)\]
for every $X,Y \in \LdpF{}$ and $D \in \Ft{t}$.  Decomposability is a stronger property than locality,
defined in~\cite{FR12} by the equality $1_D R_t(X) = 1_D R_t(1_D X)$ for every $X \in \LdpF{}$ and
$D \in \Ft{t}$, see example~\ref{notdec} below. In the scalar case both notions coincide (see e.g.~\cite{DS05}).
\end{remark}

While the typical examples of risk measures will be decomposable, we will give one example below of a risk measure that is not decomposable.

\begin{example}
\label{notdec}
Let $A := \cl\inlrparen{K + \LdiF{+}} \neq \LdiF{+}$ for some cone $K \subseteq \LdiF{0}$ such that $A$ is an acceptance set.
Let $R_t(X) = \inlrcurly{u \in M_t: X + u \in A}$ for all times $t$.  It can be shown that $\seq{R}$
is closed, convex, and multiportfolio time consistent (see the next section for details), but is not decomposable. This might be
important if one is interested in the static risk measure $R_0(X)$, sticks to the decision made at time $t=0$ and just reevaluates the
risk at time $t>0$ based on the same acceptability criterion used at $t=0$ to determine e.g. if the initial deposit $u\in R_0(X)$ can
be reduced at time $t>0$ while keeping acceptability.
Furthermore, if $M = \R^d$, then $R_t$ is local, but not decomposable.
\end{example}

%%%%%%%%%%%%%%%%%%%%%%%%%%%%%%%%%%%%%%%%
\subsection{Multiportfolio time consistency}

In \cite{FR12}  it was shown that a useful concept of time consistency for set-valued risk measures is given by a property called multiportfolio time consistency. In the following we review the definition and equivalent characterizations of this property.

\begin{definition}
\label{defn_mptc}
A dynamic risk measure $\seq{R}$ is \textbf{\emph{multiportfolio time consistent}} if for all times $0 \leq t < s \leq T$, all portfolios $X\in \LdpF{}$ and all sets ${\bf Y}\subseteq \LdpF{}$ the implication
\begin{equation*}
  R_{s}(X) \subseteq \bigcup_{Y \in {\bf Y}} R_{s}(Y) \Rightarrow R_t(X) \subseteq \bigcup_{Y \in {\bf Y}} R_t(Y)
\end{equation*}
is satisfied.
\end{definition}
The intuitive reasoning for multiportfolio time consistency is that if at some time any risk compensation portfolio for $X$ also compensates the risk of some portfolio $Y$ in the set ${\bf Y}$, then at any prior time the same relation should hold true.
\begin{theorem}[Theorem 3.4 in \cite{FR12}]
\label{thm_equiv_tc}
For a normalized dynamic risk measure $\seq{R}$ the following are equivalent:
\begin{enumerate}
\item \label{thm_equiv_tctc}$\seq{R}$ is multiportfolio time consistent,
\item \label{thm_equiv_recursive} $R_t$ is recursive, that is for all times $0 \leq t < s \leq T$
    \begin{equation}
    \label{recursive}
        R_t(X) = \bigcup_{Z \in R_{s}(X)} R_t(-Z) =: R_t(-R_{s}(X)).
    \end{equation}
\item for every time $0 \leq t < s \leq T$
    \begin{equation*}
%    \label{sum_acceptance}
        A_t = A_{t,s} + A_{s}.
    \end{equation*}
\end{enumerate}
\end{theorem}

The above theorem provides the equivalence between multiportfolio time consistency and the recursive form for set-valued risk measures.  In~\cite{FR12} it was demonstrated that the set of superhedging portfolios satisfies the recursive form, but the set-valued average value at risk does not.  Furthermore, \cite{FR12} shows that the algorithm for calculating the set of superhedging portfolios in \cite{LR11} is a result of the recursive form, and that the recursive form can be seen as a set-valued version of Bellman's principle.

In the discrete time setting $\{0,1,...,T\}$, multiportfolio time consistency is equivalent to the recursive form
using steps of size $1$ only (i.e. setting $s=t+1$ in \eqref{recursive}).

%%%%%%%%%%%%%%%%%%%%%%%%%%%%%%%%%%%%%%%%%
%%%%%%%%%%%%%%%%%%%%%%%%%%%%%%%%%%%%%%%%%
%%%%%%%%%%%%%%%%%%%%%%%%%%%%%%%%%%%%%%%%%
\section{Convex risk measures and multiportfolio time consistency}
\label{sec_convex}

In this section, we want to study the impact of multiportfolio time consistency on the penalty function of a (conditionally) convex risk measure.
In the scalar case it could be shown that (multiportfolio) time consistency is equivalent to an additive property of the penalty functions,
see e.g. \cite{FP06,CDK06,BN08,BN09,AP10}, which is called the cocycle property in~\cite{BN08,BN09}. We will show that a corresponding result is also true in the set-valued case. However, it is much harder to prove than in the scalar case. The reason is that, when following the proofs along the lines of \cite{FP06,BN09}, an additional infimum (that is the union in the recursion) appears in the set-valued case, which is not present in the scalar case. One would need to apply a minimax theorem in order to exchange the infimum and the supremum, but it is hard to verify the constraint qualification. Thus, we will follow a different route in proving the main theorem about the equivalence between multiportfolio time consistency and an additive property of the penalty functions. In the heart of this new proof lies a Hahn-Banach separation argument, which we will provide before presenting the main theorem.

The Hahn-Banach argument uses the functions $\FQwblank: \LdpF{} \to 2^{M_t}$ defined by
\begin{align*}
    \FQw{X} :=  \lrcurly{u \in M_t: \E{\trans{w}\EQt{X}{t}} \leq \E{\trans{w}u}}= \lrparen{\EQt{X}{t} + G_t(w)} \cap M_t,
\end{align*}
for $(\Q,w) \in \olW_t$.  These functions are the main ingredients in the duality theory for
set-valued functions (see \cite{H09}, example~2 and proposition~6), as they replace the continuous
linear functions used in the scalar duality theory.  Clearly, the functions $\FQwblank$ appear in
the dual representation \eqref{convex_dual} of risk measures and in the definition of the minimal
penalty function \eqref{min penalty}.

We are now ready to formulate the Hahn-Banach argument, which characterizes when a portfolio is acceptable.
\begin{lemma}
\label{separation}
Let $A_t \subseteq \LdpF{}$ be a closed convex acceptance set and let $X \in \LdpF{}$. Then, $X \not\in A_t$ if and only if there exists a $(\Q,w) \in \olW_0$ such that \[\FQwz{X} \supsetneq \cl \bigcup_{Z \in A_t} \FQwz{Z}.\]
\end{lemma}
\begin{proof}
If $X \not\in A_t$ then there exists a $Y \in \LdqF{+}$ with $\inE{\trans{Y}X} < \inf_{Z \in A_t} \inE{\trans{Y}Z}$ by the separating
hyperplane theorem.  (If we choose $Y \not\in \LdqF{+}$ then $\inf_{Z \in A_t} \inE{\trans{Y}Z} = -\infty$ by $A_t + \LdpF{+} \subseteq A_t$ which leads to a contradiction.)  This implies that
\begin{align*}
\FYvz{X} := \lrcurly{u \in M: \E{\trans{Y}X} \leq \trans{v}u}\supsetneq \lrcurly{u \in M: \inf_{Z
\in A_t} \E{\trans{Y}Z} \leq \trans{v}u} = \cl \bigcup_{Z \in A_t} \FYvz{Z}
\end{align*}
for any $v \not\in \prp{M}$  since $f(u) = \trans{v}u$ is a continuous linear operator from $M$ to
$\R$.  In particular this is true for any $v \in \inlrparen{\inE{Y} + \prp{M}}\backslash \prp{M}$.
As given in lemma 4.5 in \cite{FR12} and \eqref{smallerW}, there exists a pair $(\Q,w) \in \olW_0$ such that $\FQwz{\cdot} = \FYvz{\cdot}$, therefore $\FQwz{X} \supsetneq \cl \bigcup_{Z \in A_t} \FQwz{Z}$.

If $\FQwz{X} \supsetneq \cl \bigcup_{Z \in A_t} \FQwz{Z}$ for some $(\Q,w) \in \olW_0$, then $\inEQ{X} \neq \inEQ{Z}$ for all $Z \in A_t$.  Therefore $X \not\in A_t$.
\qed\end{proof}

In order to formulate the additive property of the penalty functions, we need to define the minimal stepped penalty function $-\beta_{t,s}^{\min}$ and $-\alpha_{t,s}^{\min}$ (stepped from $t$ to $s>t$). The definition is straight forward, using the definition of minimal penalty functions \eqref{min penalty}, respectively \eqref{conditional_min penalty}, but with stepped acceptance sets.
Define $-\beta_{t,s}^{\min}$ by
\begin{equation}
\label{stepped_beta}
-\beta_{t,s}^{\min}(\Q,w) := \cl \bigcup_{X \in A_{t,s}} \lrparen{\EQt{X}{t} + G_t(w)} \cap M_t
\end{equation}
and $-\alpha_{t,s}^{\min}$ by
\begin{equation*}
%\label{stepped_alpha}
-\alpha_{t,s}^{\min}(\Q,w) := \cl \bigcup_{X \in A_{t,s}} \lrparen{\EQt{X}{t} + \Gamma_t(w)} \cap M_t
\end{equation*}
for $(\Q,w) \in \olW_{t,s} = \inlrcurly{(\Q,w) \in \mathcal{M} \times \inlrparen{\plus{M_{t,+}} \backslash M_t^{\perp}}: w_t^{s}(\Q,w) \in \plus{M_{s,+}}}$. A detailed discussion about stepped risk measures can be found in section~\ref{sec_stepped} of the appendix.

We now state the main results of this section. Its proofs are based on the Hahn-Banach argument given above and several lemmas provided in the appendix, sections~\ref{sec_dualvar} and \ref{sec_sum_acceptance}, that concern e.g. the relation of dual variables at different times.
Throughout the remainder of this paper we will use the notation $\Q^{s}$ to denote the modification of $\Q \in \mathcal{M}$ defined by $\frac{d\Q^{s}}{d\P} = \xi_{s,T}(\Q)$.

\begin{theorem}
\label{thm_mptc_penalty_wo_uc}
Let $\seq{R}$ be a dynamic normalized closed convex risk measure. Then $\seq{R}$ is multiportfolio time consistent if and only if  for every $(\Q,w) \in \olW_t$
\[-\beta_{t}^{\min}(\Q,w) = \cl \lrparen{-\beta_{t,s}^{\min}(\Q,w) + \EQt{-\beta_{s}^{\min}(\Q^{s},w_t^{s}(\Q,w))}{t}}\]
and $A_{t,s} + A_{s}$ is closed, for all $0 \leq t < s \leq T$.
\end{theorem}

\begin{proof}
From theorem~\ref{thm_equiv_tc}, a normalized dynamic risk measure
is multiportfolio time consistent if and only if $A_t = A_{t,s} + A_{s}$ for every $0 \leq t < s \leq T$.

\begin{proofenum}
\item Assume $\seq{R}$ is a normalized closed convex multiportfolio time consistent risk measure, i.e. assume $A_t = A_{t,s} + A_{s}$. It immediately follows that for any $(\Q,w) \in \olW_t$
    \begin{align}
        \nonumber & \cl \lrparen{-\beta_{t,s}^{\min}(\Q,w) + \EQt{-\beta_{s}^{\min}(\Q^{s},w_t^{s}(\Q,w))}{t}}\\
        \label{exp_penalty} &= \cl \lrparen{-\beta_{t,s}^{\min}(\Q,w) + \cl \bigcup_{X_{s} \in A_{s}} \lrparen{\EQt{X_{s}}{t} + G_t(w)} \cap M_t}\\
        \nonumber &= \cl \lparen{\cl \bigcup_{X_{t,s} \in A_{t,s}} \lrparen{\EQt{X_{t,s}}{t} + G_t(w)} \cap M_t} \rparen{+ \cl \bigcup_{X_{s} \in A_{s}} \lrparen{\EQt{X_{s}}{t} + G_t(w)} \cap M_t}\\
        \label{closure_sum_closures} &= \cl \lparen{\bigcup_{X_{t,s} \in A_{t,s}} \lrparen{\EQt{X_{t,s}}{t} + G_t(w)} \cap M_t}\rparen{+ \bigcup_{X_{s} \in A_{s}} \lrparen{\EQt{X_{s}}{t} + G_t(w)} \cap M_t}\\
        \nonumber &= \cl \bigcup_{\substack{X_{t,s} \in A_{t,s}\\ X_{s} \in A_{s}}} \lrparen{\EQt{X_{t,s} + X_{s}}{t} + G_t(w)} \cap M_t\\
        \nonumber &= \cl \bigcup_{X \in A_t} \lrparen{\EQt{X}{t} + G_t(w)} \cap M_t = -\beta_t^{\min}(\Q,w).
    \end{align}
    Equation~\eqref{exp_penalty} follows from lemma~\ref{lemma_exp_penalty}, and equation~\eqref{closure_sum_closures} follows from proposition 1.23 in~\cite{L11}.  Note that if $(\Q,w) \in \olW_t$ then $(\Q,w) \in \olW_{t,s}$ (see remark~\ref{rem_stepped_dual_var}) and $(\Q^{s},w_t^{s}(\Q,w)) \in \olW_{s}$ (see lemma~\ref{lemma_Wtau_to_Wt}).

\item Conversely, assume $A_{t,s} + A_{s}$ is closed and the cocycle condition is satisfied, that is, $-\beta_{t}^{\min}(\Q,w) = \cl
\inlrparen{-\beta_{t,s}^{\min}(\Q,w) + \inEQt{-\beta_{s}^{\min}(\Q^{s},w_t^{s}(\Q,w))}{t}}$ for every $(\Q,w) \in \olW_t$.

    Note that for any $(\Q,w) \in \olW_{0}$ it holds $w_{0}^{s}(\Q,w) = w_t^{s}(\Q^t,w_{0}^t(\Q,w))$.

    Let $X \in A_{t,s} + A_{s}$, then
    by the tower property and corollary~\ref{prop_Gw}, for every $(\Q,w) \in \olW_0$
    \begin{align*}
        \FQwz{X} & \subseteq \EQ{\cl \lrparen{-\beta_{t,s}^{\min}(\Q^t,w_0^t(\Q,w)) + \EQt{-\beta_{s}^{\min}(\Q^{s},w_0^{s}(\Q,w))}{t}}}\\
        & = \EQ{-\beta_t^{\min}(\Q^t,w_0^t(\Q,w))} = \cl \bigcup_{Z \in A_t} \FQwz{Z}.
    \end{align*}
    The last equality follows from lemma~\ref{lemma_exp_penalty}.
    If $X \not\in A_t$ then, by lemma~\ref{separation}, there exists a pair $(\Q,w) \in \olW_0$ such that $\FQwz{X} \supsetneq \cl \bigcup_{Z \in A_t} \FQwz{Z}$. However, this is a contradiction to the above, therefore $X \in A_t$ and thus
    \begin{equation}
    \label{resproof}
    A_{t,s} + A_{s}\subseteq A_t.
    \end{equation}

    Let $X \in A_t$, then (using corollary~\ref{prop_Gw} and lemma~\ref{lemma_exp_penalty}) for every $(\Q,w) \in \olW_0$
    \begin{align*}
        \FQwz{X} & \subseteq \EQ{-\beta_t^{\min}(\Q^t,w_0^t(\Q,w))}\\
        & = \EQ{\cl \lrparen{-\beta_{t,s}^{\min}(\Q^t,w_0^t(\Q,w)) + \EQt{-\beta_{s}^{\min}(\Q^{s},w_0^{s}(\Q,w))}{t}}}\\
        & = \cl \bigcup_{Z \in A_{t,s} + A_{s}} \FQwz{Z}.
    \end{align*}
    If we assume that $X \not\in A_{t,s} + A_{s}$ (which is closed by assumption and is a convex acceptance set by lemma~\ref{lemma_sum_acceptance}, where the assumption $A_{t,s} + A_{s} \subseteq A_t$ of lemma~\ref{lemma_sum_acceptance} is satisfied by \eqref{resproof})
    then, by lemma~\ref{separation}, there exists a pair $(\Q,w) \in \olW_0$ such that \[\FQwz{X} \supsetneq \cl \bigcup_{Z \in A_{t,s} + A_{s}} \FQwz{Z}.\]  This is a contradiction to the above, therefore $X \in A_{t,s} + A_{s}$.
\end{proofenum}
\qed\end{proof}

\begin{corollary}
\label{thm_mptc_penalty}
Let $\seq{R}$ be a dynamic normalized c.u.c. convex risk measure. Then, $\seq{R}$ is multiportfolio time consistent if and only if
\begin{equation*}
-\beta_{t}^{\min}(\Q,w) = \cl \lrparen{-\beta_{t,s}^{\min}(\Q,w) + \EQt{-\beta_{s}^{\min}(\Q^{s},w_t^{s}(\Q,w))}{t}}
\end{equation*}
holds for every $(\Q,w) \in \olW_t$ for all $0 \leq t < s \leq T$.
\end{corollary}

\begin{proof}
In light of theorem~\ref{thm_mptc_penalty_wo_uc} it only remains to show that convex upper continuity of $\seq{R}$ implies the closedness of $A_{t,s} + A_{s}$. This follows from remark~\ref{rem_stepped_usc} and lemma~\ref{lemma_closed_sum_acceptance}.
\qed\end{proof}

In the above theorem and corollary, we have demonstrated the equivalence between an additive property for the penalty functions and multiportfolio time consistency.  This allows us to define risk measures by the penalty functions alone and verify whether the corresponding c.u.c. convex risk measure is multiportfolio time consistent.

\begin{example}[Entropic risk measure]
\label{ex_entropic}
The restrictive entropic risk measure  (see section~\ref{sec_entropic} for a more general and detailed treatment) is defined by
\[
R_t^{ent}(X;\lambda) := \lrcurly{u \in \LdiF{t}: \Et{u_t(X + u)}{t} \succeq 0}
\]
for every $X \in \LdiF{}$ where
$u_t(x) =\transp{u_{t,1}(x_1),...,u_{t,d}(x_d)}$ for any $x \in \R^d$ with elements $u_{t,i}(z)=\frac{1-e^{-\lambda_i z}}{\lambda_i}$ for $z\in\R$ and $i=1,...,d$.

The restrictive entropic risk measure is normalized, convex, and closed. By applying lemma~\ref{lemma_entropic_dual} one obtains that its stepped penalty functions, defined in \eqref{stepped_beta} (with $0 \leq t < s \leq T$), are given by
\begin{equation*}
-\beta_{t,s}^{ent}(\Q,w;\lambda) := -\frac{\hat{H}_{t,s}(\Q|\P)}{\lambda} + G_t(w)
\end{equation*}
for any $(\Q,w) \in \olW_t$, where $\hat{H}_{t,s}(\Q|\P) := \inEQt{\log(\xi_{t,s}(\Q))}{t}$.
The dual representation of the entropic risk measure is given by \eqref{convex_dual} with minimal penalty function $-\beta_t^{ent} := -\beta_{t,T}^{ent}$. Note that
$\hat{H}_{t,T}(\Q|\P) = \inEQt{\log(\dQdP)}{t}$ is the conditional relative entropy.

It can immediately be seen that
\begin{align*}
	-\beta_t^{ent}(\Q,w;\lambda) &= \cl\lrparen{-\beta_{t,s}(\Q,w;\lambda)+ \EQt{-\beta_{s}^{ent}(\Q^{s},w_t^{s}(\Q,w);\lambda)}{t}}
\end{align*}
for all times $0 \leq t < s \leq T$ and any dual variables $(\Q,w) \in \olW_t$. Thus, the restrictive entropic risk
measure satisfies the cocycle property. Proposition~\ref{prop_uc_scalar} yields (using representation \eqref{entpoint} for the restrictive
entropic risk measure) the convex upper continuity of $R_t^{ent}$.  Therefore the set-valued entropic risk measure with constant risk aversion parameter $\lambda$ and restrictive thresholds $\LdiF{t,+}$ is multiportfolio time consistent by corollary~\ref{thm_mptc_penalty}.
\end{example}

We now consider the conditionally convex case when a dual representation w.r.t. equivalent probability measures, i.e. w.r.t. the dual set $\olW_t^e$, holds.
\begin{corollary}
\label{cor_mptc_cond_penalty}
Let $\seq{R}$ be a dynamic normalized closed conditionally convex risk measure with dual representation
\begin{equation}
\label{dual_We}
R_t(X) = \bigcap_{(\Q,w) \in \olW_t^e} \lrsquare{-\alpha_t^{\min}(\Q,w) + \lrparen{\EQt{-X}{t} + \Gamma_t(w)} \cap M_t}
\end{equation}
for every $X \in \LdpF{T}$ where $\olW_t^e = \inlrcurly{(\Q,w) \in \olW_t: \Q \in \mathcal{M}^e}$.  Then $\seq{R}$ is multiportfolio time consistent if and only if  for every $(\Q,w) \in \olW_t$
\[-\alpha_{t}^{\min}(\Q,w) = \cl \lrparen{-\alpha_{t,s}^{\min}(\Q,w) + \EQt{-\alpha_{s}^{\min}(\Q^{s},w_t^{s}(\Q,w))}{t}}\]
and $A_{t,s} + A_{s}$ is closed, for all $0 \leq t < s \leq T$.
\end{corollary}
\begin{proof}
This follows using the same logic as in the proof of theorem~\ref{thm_mptc_penalty_wo_uc} (with the closure operator added where necessary) using the results on the conditional expectation of $\alpha_s^{\min}$ and $\Gamma_s$ in lemma~\ref{lemma_exp_cond_penalty} and corollary~\ref{prop_Gamma}.
\qed\end{proof}
If $\seq{R}$ is, additionally to the assumptions of Corollary~\ref{cor_mptc_cond_penalty}, conditionally c.u.c., then $\seq{R}$ is multiportfolio time consistent if and only if the cocycle condition on
$\alpha_t^{\min}$ holds. This corresponds to the results of corollary~\ref{thm_mptc_penalty}.

%%%%%%%%%%%%%%%%%%%%%%%%%%%%%%%%%%
%%%%%%%%%%%%%%%%%%%%%%%%%%%%%%%%%%%%
\section{Coherent risk measures and multiportfolio time consistency}
\label{sec_coherent}
In this section, we want to study multiportfolio time consistency in the (conditionally) coherent case. In particular, we want to find equivalent characterizations of multiportfolio time consistency with respect to the set of dual variables.
In the scalar framework an equivalent property is given by stability of the dual variables, also called m-stability, which was studied for the case when the dual probability measures are absolutely continuous to the real world probability measure $\P$ in~\cite{AP10,CDK06}, and when the dual probability measures are equivalent to $\P$ in~\cite{D06,FP06,AD07}.

\begin{remark}
\label{rem_cuc}
In this section and section~\ref{sec_composition}, we will for simplicity only present the results for multiportfolio time consistency assuming
convex upper continuity, akin to corollary~\ref{thm_mptc_penalty} above.  The results can be given
for closed risk measures as well as it was done in theorem~\ref{thm_mptc_penalty_wo_uc} and corollary \ref{cor_mptc_cond_penalty}.
\end{remark}

For the results below we use the definition of the maximal set of stepped dual variables $\olW_{t,s}^{\max} \subseteq \olW_{t,s}$ as defined in section~\ref{sec_stepped}.  That is,
\begin{align*}
\olW_{t,s}^{\max} &= \{(\Q,w) \in \olW_{t,s}: -\beta_{t,s}^{\min}(\Q,w) = G_t(w) \cap M_t\}= \lrcurly{(\Q,w) \in \olW_{t,s}: w_t^{s}(\Q,w) \in \plus{A_{t,s}}}.
\end{align*}
All the results for the conditionally coherent case stay the same as for the coherent case (except that the assumption c.u.c. can be weakened to conditionally c.u.c.)
as the set of dual variables does not change (compare \eqref{conditional_convex_dual} and \eqref{convex_dual}). This is also true for
the maximal set of stepped dual variables as \[\olW_{t,s}^{\max}=\{(\Q,w) \in \olW_{t,s}: -\alpha_{t,s}^{\min}(\Q,w) = \Gamma_t(w) \cap
M_t\}\]since $\plus{A_{t,s}}=\inlrcurly{v \in \LdqF{t}: \forall u \in A_{t,s}, \trans{v}u \geq 0\; \Pas}$ if $R_t$ is conditionally coherent.

\begin{remark}
\label{rem_stepped_dual_max}
For any closed coherent risk measure $R_t$ (not necessarily multiportfolio time consistent) it can trivially be seen that $\olW_t^{\max} \subseteq \olW_{t,s}^{\max}$ since $\olW_t \subseteq \olW_{t,s}$ (see remark~\ref{rem_stepped_dual_var}) and $w_t^{T}(\Q,w) \in \plus{A_t}$ implies $w_t^{s}(\Q,w) \in \plus{A_{t,s}}$.
\end{remark}

The first result we provide, which will be useful for generating a c.u.c. coherent multiportfolio time consistent risk measure in section~\ref{sec_composition}, is a corollary to theorem~~\ref{thm_mptc_penalty_wo_uc} (respectively corollary~\ref{thm_mptc_penalty}) above.

Let us define the set $H_t^{s}: 2^{\olW_{s}} \to 2^{\olW_t}$ by
\begin{equation*}
H_t^{s}(D) := \lrcurly{(\Q,w) \in \olW_t: \lrparen{\Q^{s},w_t^{s}(\Q,w)} \in D}
\end{equation*}
for any $0 \leq t < s \leq T$, and any $D \subseteq \olW_{s}$.

\begin{corollary}
\label{cor_mptc_coherent}
Let $\seq{R}$ be a dynamic normalized c.u.c. coherent risk measure. Then, $\seq{R}$ is multiportfolio time consistent if and only if for all $0 \leq t < s \leq T$ it holds \[\olW_t^{\max} = \olW_{t,s}^{\max} \cap H_t^{s}\lrparen{\olW_{s}^{\max}}.\]
\end{corollary}
\begin{proof}
This follows trivially from corollary~\ref{thm_mptc_penalty} and corollary~\ref{prop_Gw} by noting that for any times $t$ and $s > t$
\begin{align*}
-\beta_t^{\min}(\Q,w) & = \begin{cases} G_t(w) \cap M_t & \text{if }(\Q,w) \in \olW_{t}^{\max} \\
M_t & \text{else}
\end{cases}\\
-\beta_{t,s}^{\min}(\Q,w) & = \begin{cases} G_t(w) \cap M_t & \text{if }(\Q,w) \in \olW_{t,s}^{\max} \\
M_t & \text{else}.
\end{cases}
\end{align*}
And since $\olW_{t,s} \supseteq \olW_t$ (see remark~\ref{rem_stepped_dual_var}) for any times $t < s$, the result follows.
\qed\end{proof}

We now want to study the pasting of dual variables and the generalization of stability to the set-valued case.

For $\Q,\R \in \mathcal{M}$ we denote by $\Q \oplus^{s} \R$ the pasting of $\Q$ and $\R$ at $s$, i.e. the vector probability measures $\S\in \mathcal{M}$ defined via
\[\dSdP = \diag{\xi_{0,s}(\Q)} \xi_{s,T}(\R).\]
Note that, if $\S = \Q \oplus^{s} \R$, $\bar{\xi}_{t,r}(\S_i) = \bar{\xi}_{t,r}(\Q_i)$ for  $t \leq r \leq s$, but $\bar{\xi}_{t,r}(\S_i)$
is not necessarily equal to $\bar{\xi}_{t,r}(\R_i)$ for $r \geq t > s$.  If $\Q = \P|_{\Ft{t}}$ for some $t \leq s$ (i.e.
$\bar{\xi}_{0,t}(\Q_i) = 1$ almost surely for every $i \in \{1,...,d\}$), then it follows that $w_t^r(\S,w) = w_{s}^r(\R,w_t^{s}(\Q,w))$ for $0 \leq t \leq s \leq r \leq T$ and any $w \in \LdqF{t}$.
In the set-valued framework we will define stability as a property with respect to two other sets. This is due to the fact that our dual variables consists of pairs. Naturally, stability is a property that imposes conditions on both components of a pair $(\Q,w)$.

\begin{definition}
\label{defn_stable}
A set $W_t \subseteq \olW_t$ is called \textbf{\emph{stable}} at time $t$ with respect to $W_{t,s}$ and $W_{s}$ for $s > t$ if
\begin{enumerate}
    \item $(\Q,w) \in W_t$ implies $(\Q^{s},w_t^{s}(\Q,w)) \in W_{s}$ and
    \item $(\Q,w) \in W_{t,s}$ and $\R \in \mathcal{M}$ such that $(\R,w_t^{s}(\Q,w)) \in W_{s}$ implies $(\Q \oplus^{s} \R,w) \in W_t$.
\end{enumerate}
\end{definition}

\begin{remark}
In the scalar framework, stability is defined with respect to stopping times, see e.g.~\cite{D06}.  We are able to weaken this
assumption in the set-valued framework due to the total ordering given by the half-space $G_t(w)$ generated by the second dual
variable, see lemma~\ref{lemma_exp_penalty} for more details.
\end{remark}

The main theorem of this section is given below.  It provides an equivalence between the stability of the sets of dual variables $\olW_t^{\max}$ and multiportfolio time consistency.  We present an additional property which is equivalent to stability and therefore to multiportfolio time consistency.  This additional property, given in equation~\eqref{eq_stable}, is a generalization of property (2) of corollary 1.26 from~\cite{AP10}.

\begin{theorem}
\label{thm_mptc_stable}
Let $\seq{R}$ be a normalized c.u.c. coherent risk measure, then the following are equivalent:
\begin{enumerate}
\item $\seq{R}$ is multiportfolio time consistent;
\item $\olW_t^{\max}$ is stable at time $t$ with respect to $\olW_{t,s}^{\max}$ and $\olW_{s}^{\max}$ for every $0 \leq t < s \leq T$;
\item for every time $0 \leq t < s \leq T$
\begin{equation}
\label{eq_stable}
\olW_t^{\max} = \lrcurly{(\Q \oplus^{s} \R,w): (\Q,w) \in \olW_{t,s}^{\max}, (\R,w_t^{s}(\Q,w)) \in \olW_{s}^{\max}}.
\end{equation}
\end{enumerate}
\end{theorem}
\begin{proof}
We will show that multiportfolio time consistency implies stability, stability implies equation~\eqref{eq_stable}, and finally, that equation~\eqref{eq_stable} implies multiportfolio time consistency.

\begin{proofenum}
\item[1.$\Rightarrow$2.] Assume $\seq{R}$ is multiportfolio time consistent.  We want to show that $\olW_t^{\max}$ is stable at time $t$ with respect to $\olW_{t,s}^{\max}$ and $\olW_{s}^{\max}$, as given in definition~\ref{defn_stable}.
    \begin{proofenum}
    \item By corollary~\ref{cor_mptc_coherent} it follows that $\olW_t^{\max} \subseteq H_{t}^{s}(\olW_{s}^{\max})$ and thus $(\Q,w) \in \olW_t^{\max}$ implies $(\Q^{s},w_t^{s}(\Q,w)) \in \olW_{s}^{\max}$.
    \item
        Let $t<s$, $(\Q,w) \in \olW_{t,s}^{\max}$ and $\R \in \mathcal{M}$  with $(\R,w_t^{s}(\Q,w)) \in \olW_{s}^{\max}$. We need to show that $(\S,w) \in \olW_t^{\max}$ where $\S=\Q \oplus^{s} \R$ is the pasting of $\Q$ and $\R$ at time $s$.
        \begin{proofenum}
        \item For any index $i = 1,...,d$, $\S_i = \P|_{\Ft{t}}$ since $\inEt{\dSidP}{t} = \inEt{\dQidP}{t} = 1$ almost surely.
        \item $(\S,w) \in \olW_t$ since $\S \in \mathcal{M}$, $w \in \plus{M_{t,+}} \backslash \prp{M_t}$, and
            \[w_t^T(\S,w) = w_{s}^T(\R,w_t^{s}(\Q,w)) \in \LdqF{+}.\]
        \item $(\S,w) \in \olW_t^{\max}$ if $-\beta_t^{\min}(\S,w) = G_t(w) \cap M_t$. By corollary~\ref{thm_mptc_penalty} it follows that for any $(\S,w) \in \olW_t$
        \[-\beta_t^{\min}(\S,w) = \cl\lrparen{-\beta_{t,s}^{\min}(\S,w) + \ESt{-\beta_{s}^{\min}(\S^{s},w_t^{s}(\S,w))}{t}}.\]
       We can see that $-\beta_{t,s}^{\min}(\S,w) = -\beta_{t,s}^{\min}(\Q,w)$ by the tower
property.  One can show that $-\beta_{s}^{\min}(\S^{s},w_t^{s}(\S,w)) = -\beta_{s}^{\min}(\R,w_t^{s}(\Q,w))$ by
            \begin{align*}
            -\beta_{s}^{\min}(\S^{s},w_t^{s}(\S,w))& = \cl \bigcup_{X \in A_{s}} \lrparen{\ESt{X}{s} + G_{s}(w_t^{s}(\S,w))} \cap M_{s}\\
            & = \cl \bigcup_{X \in A_{s}} \lrcurly{x \in M_{s}: \E{\trans{w_t^{s}(\S,w)}\ESt{X}{s}} \leq \E{\trans{w_t^{s}(\S,w)}x}}\\
            & = \cl \bigcup_{X \in A_{s}} \lrcurly{x \in M_{s}: \E{\trans{w_t^{T}(\S,w)}X} \leq \E{\trans{w_t^{s}(\Q,w)}x}}\\
            & = \cl \bigcup_{X \in A_{s}} \lrcurly{x \in M_{s}: \E{\trans{w_{s}^T(\R,w_t^{s}(\Q,w))}X} \leq \E{\trans{w_t^{s}(\Q,w)}x}}\\
            & = \cl \bigcup_{X \in A_{s}} \lrcurly{x \in M_{s}: \E{\trans{w_t^{s}(\Q,w)}\ERt{X}{s}} \leq \E{\trans{w_t^{s}(\Q,w)}x}}\\
            & = -\beta_{s}^{\min}(\R,w_t^{s}(\Q,w)).
            \end{align*}
            Therefore,
            \begin{align*}
            -\beta_t^{\min}(\S,w) & = \cl\lrparen{-\beta_{t,s}^{\min}(\S,w) + \ESt{-\beta_{s}^{\min}(\S^{s},w_t^{s}(\S,w))}{t}}\\
            & = \cl\lrparen{-\beta_{t,s}^{\min}(\Q,w) + \EQt{-\beta_{s}^{\min}(\R,w_t^{s}(\Q,w))}{t}}\\
            & = \cl\lrparen{G_t(w) \cap M_t + \EQt{G_{s}(w_t^{s}(\Q,w)) \cap M_{s}}{t}}
            = G_t(w) \cap M_t,
            \end{align*}
             using $(\Q,w) \in \olW_{t,s}^{\max}$ and $(\R,w_t^{s}(\Q,w)) \in \olW_{s}^{\max}$. The last line follows from corollary~\ref{prop_Gw}.
             \end{proofenum}
    \end{proofenum}
    Therefore for every time $t$, $\olW_t^{\max}$ is stable at time $t$ with respect to $\olW_{t,s}^{\max}$ and $\olW_{s}^{\max}$ for every $s > t$.

\item[2.$\Rightarrow$3.] We will demonstrate that stability implies equation~\eqref{eq_stable}. If for every time $t$, $\olW_t^{\max}$ is stable at time $t$ with respect to $\olW_{t,s}^{\max}$ and $\olW_{s}^{\max}$ for every $s > t$ then trivially ``$\supseteq$" in equation~\eqref{eq_stable} follows by the second property of stability.  By the first property of stability and remark~\ref{rem_stepped_dual_max}, for any $(\Q,w) \in \olW_t^{\max}$ it follows that $(\Q,w) \in \olW_{t,s}^{\max}$ and $(\Q^{s},w_t^{s}(\Q,w)) \in \olW_{s}^{\max}$.  Since $\Q = \Q \oplus^{s} \Q^{s}$ for any time $s$ and any probability measure $\Q \in \mathcal{M}$, then ``$\subseteq$" in equation~\eqref{eq_stable} trivially follows.

\item[3.$\Rightarrow$1.] We will prove that equation~\eqref{eq_stable} implies that for every $(\Q,w) \in
\olW_t$
\[-\beta_t^{\min}(\Q,w) = \cl \lrparen{-\beta_{t,s}^{\min}(\Q,w) +
\EQt{-\beta_{s}^{\min}(\Q^{s},w_t^{s}(\Q,w))}{t}}\]
which in turn implies multiportfolio time consistency by corollary~\ref{thm_mptc_penalty}.  We will define the set
$\widetilde{\W}_t^{s} := \inlrcurly{(\Q \oplus^{s} \R,w): (\Q,w) \in \olW_{t,s}^{\max}, (\R,w_t^{s}(\Q,w)) \in \olW_{s}^{\max}}$ for notational purposes.
    \begin{proofenum}
    \item We will show that the inclusion $\olW_t^{\max} \subseteq \widetilde{\W}_t^{s}$ implies the penalty function inclusion
$-\beta_t^{\min}(\Q,w) \supseteq \cl\inlrparen{-\beta_{t,s}^{\min}(\Q,w) + \inEQt{-\beta_{s}^{\min}(\Q^{s},w_t^{s}(\Q,w))}{t}}$ for every $(\Q,w) \in \olW_t$.
        \begin{proofenum}
        \item Let $(\S,w) \in \olW_t^{\max}$. Then, $-\beta_t^{\min}(\S,w) = G_t(w) \cap M_t$.  Additionally, there exists a $\Q,\R \in \mathcal{M}$ such that $(\Q,w) \in \olW_{t,s}^{\max}$, $(\R,w_t^{s}(\Q,w)) \in \olW_{s}^{\max}$, and $\S = \Q \oplus^{s} \R$.
This implies \[-\beta_{t,s}^{\min}(\S,w) = -\beta_{t,s}^{\min}(\Q,w) = G_t(w) \cap M_t\] and
\begin{align*}
-\beta_{s}^{\min}(\S^{s},w_t^{s}(\S,w)) &= -\beta_{s}^{\min}(\R,w_t^{s}(\Q,w))= G_{s}(w_t^{s}(\Q,w)) \cap M_{s} = G_{s}(w_t^{s}(\S,w)) \cap M_{s}.
\end{align*}
Therefore, corollary~\ref{prop_Gw} yields
\begin{align*}
-\beta_t^{\min}(\S,w) &= G_t(w) \cap M_t= \cl\lrparen{G_t(w) \cap M_t + \ESt{G_{s}(w_t^{s}(\S,w)) \cap M_{s}}{t}}\\
&= \cl\lrparen{-\beta_{t,s}^{\min}(\S,w) + \ESt{-\beta_{s}^{\min}(\S^{s},w_t^{s}(\S,w))}{t}}.
 \end{align*}
        \item Let $(\S,w) \in \olW_t \backslash \olW_t^{\max}$, then $-\beta_t^{\min}(\S,w) = M_t$, and thus \[-\beta_t^{\min}(\S,w) \supseteq \cl\lrparen{-\beta_{t,s}^{\min}(\S,w) + \ESt{-\beta_{s}^{\min}(\S^{s},w_t^{s}(\S,w))}{t}}.\]
        \end{proofenum}
    \item We will show that the inclusion $\olW_t^{\max} \supseteq \widetilde{\W}_t^{s}$ implies the penalty function inclusion
$-\beta_t^{\min}(\Q,w) \subseteq \cl\inlrparen{-\beta_{t,s}^{\min}(\Q,w) + \inEQt{-\beta_{s}^{\min}(\Q^{s},w_t^{s}(\Q,w))}{t}}$ for every $(\Q,w) \in \olW_t$.
        \begin{proofenum}
        \item Let $(\S,w) \in \widetilde{\W}_t^{s}$. By the assumption, we have $-\beta_t^{\min}(\S,w) = G_t(w) \cap M_t$ and $(\S,w) \in \olW_t^{\max}$.  Additionally there exists $\Q,\R \in \mathcal{M}$
such that $(\Q,w) \in \olW_{t,s}^{\max}$, $(\R,w_t^{s}(\Q,w)) \in \olW_{s}^{\max}$, and
$\S = \Q \oplus^{s} \R$.  This implies the penalty functions are half-spaces, i.e. $-\beta_{t,s}^{\min}(\S,w) = G_t(w) \cap M_t$
and $-\beta_{s}^{\min}(\S^{s},w_t^{s}(\S,w)) = G_{s}(w_t^{s}(\S,w)) \cap
M_{s}$. Therefore we have the equality \[-\beta_t^{\min}(\S,w) = \cl\lrparen{-\beta_{t,s}^{\min}(\S,w) + \ESt{-\beta_{s}^{\min}(\S^{s},w_t^{s}(\S,w))}{t}}.\]
        \item Let $(\S,w) \in \olW_t \backslash \widetilde{\W}_t^{s}$, i.e. for all probability measures $\Q,\R \in \mathcal{M}$ such that $\S = \Q
\oplus^{s} \R = \Q \oplus^{s} \R^{s}$ either $(\Q,w) \not\in \olW_{t,s}^{\max}$ or $(\R^{s},w_t^{s}(\Q,w)) \not\in \olW_{s}^{\max}$.
This implies for any $\Q,\R \in \mathcal{M}$ where $\S = \Q \oplus^{s} \R$ either
\[-\beta_{t,s}^{\min}(\S,w) = -\beta_{t,s}^{\min}(\Q,w) = M_t\] or \[-\beta_{s}^{\min}(\S^{s},w_t^{s}(\S,w)) =
-\beta_{s}^{\min}(\R^{s},w_t^{s}(\Q,w)) = M_{s}.\]
Therefore we have that $\cl\inlrparen{-\beta_{t,s}^{\min}(\S,w) + \inESt{-\beta_{s}^{\min}(\S^{s},w_t^{s}(\S,w))}{t}} = M_t$,
and thus
\[-\beta_t^{\min}(\S,w) \subseteq \cl\inlrparen{-\beta_{t,s}^{\min}(\S,w) +
\inESt{-\beta_{s}^{\min}(\S^{s},w_t^{s}(\S,w))}{t}}.\]
        \end{proofenum}
    \end{proofenum}
\end{proofenum}
\qed\end{proof}

The above theorem provides two equivalent representations for multiportfolio time consistency for coherent risk measures.  This generalizes the stability property for scalar risk measures, which is a well known result.  Conceptually, stability means that pasting together dual variables
creates another possible dual variable, which logically corresponds with time consistency concepts.

We conclude this section by providing as an example the set of superhedging portfolios, for which the set of dual variables satisfies stability.
\begin{example}[Superhedging price]
\label{ex_shp}
Consider the discrete time setting with $M_t = \LdpF{t}$
for all times $t \in \{0,1,...,T\}$.  Consider a market with proportional transaction costs as in \cite{K99,S04,KS09}, which is modeled by a sequence of solvency cones $\seq{K}$. $K_t$ is a solvency cone at time $t$ if it is an $\Ft{t}$-measurable cone such that for every $\omega \in \Omega$, $K_t[\omega]$ is a closed convex cone with $\mathbb{R}_+^d \subseteq K_t[\omega] \subsetneq \mathbb{R}^d$.

Modifying the representation from ~\cite{FR12} by using corollary~\ref{cor_conditional_dual} shows that the set of superhedging portfolios has the following dual representation (under the robust no arbitrage condition, see \cite{FR12})
\begin{equation*}
SHP_t(X) = \bigcap_{(\Q,w) \in \olW_{\lrcurly{t,...,T}}} \lrparen{\EQt{X}{t} + \Gamma_t(w)},
\end{equation*}
where $t,s \in\{0,1,...,T\}$ with $t < s$ and
\begin{align*}
\olW_{\lrcurly{t,...,s}} & := \Big\{(\Q,w) \in \olW_{t,s}: w_t^{r}(\Q,w) \in \LdpK{q}{r}{\plus{K_{r}}} \; \forall r \in \lrcurly{t,...,s}, w_t^{s}(\Q,w) \in \plusp{\LdpF{s} \cap \sum_{r = s + 1}^T \LdpK{p}{r}{K_{r}}}\Big\}.
\end{align*}
It was shown in~\cite{FR12} that the set-valued function given by $R_t(X) := SHP_t(-X)$ defines a dynamic risk measure
which is normalized, closed, conditionally coherent, and multiportfolio time consistent.  Its acceptance set and stepped acceptance set are given by
$A_t= \sum_{s = t}^{T} \LdpK{p}{s}{K_s}$  and $A_{t,s} = A_t \cap \LdpF{s}$, respectively.
Closure and multiportfolio time consistency imply that
the sum $A_{t,s} + A_{s}$ is closed, and therefore convex upper continuity is not needed in theorem~\ref{thm_mptc_stable} (see also remark~\ref{rem_cuc}).
Then, stability of the dual set $\olW_{\inlrcurly{t,...,T}}$ follows from theorem~\ref{thm_mptc_stable}.
Alternatively, one can directly prove that for any time $t$ and any $s > t$,
\[\olW_{\lrcurly{t,...,T}} = \lrcurly{(\Q \oplus^{s} \R,w): (\Q,w) \in \olW_{\lrcurly{t,...,s}}, (\R,w_t^{s}(\Q,w)) \in \olW_{\lrcurly{s,...,T}}}\]
holds, which is by theorem~\ref{thm_mptc_stable} equivalent to stability.
\end{example}

%%%%%%%%%%%%%%%%%%%%%%%%%%%%%%%%
\section{Composition of one-step risk measures}
\label{sec_composition}

For this section we will restrict ourselves to the discrete time setting $t \in \{0,1,...,T\}$.
As in section~2.1 in \cite{CK10} and section~4 in \cite{CS09}, a (multiportfolio) time consistent version of any scalar dynamic risk measure
can be created through backwards recursion.  In the following we recall the corresponding results from proposition~3.11 and corollary~3.14
in~\cite{FR12} in the set-valued framework.  Then, in corollary~\ref{cor_composed}, we prove an equivalent formulation for c.u.c. convex and
coherent risk measures, which will be very useful to deduce dual representations of composed (and thus multiportfolio time consistent) risks
measures.

\begin{proposition}
\label{prop_gen_tc}
Let $\seq{R}$ be a dynamic risk measure on $\LdpF{}$,
then $\seq{\tilde{R}}$ defined for all $X \in \LdpF{}$ by
\begin{align}
\label{eqn_composed_final} \tilde{R}_{T}(X) & = R_{T}(X),\\
\label{eqn_composed} \forall t \in \lrcurly{0,1,...,T-1}: \; \tilde{R}_t(X) & = \bigcup_{Z \in \tilde{R}_{t+1}(X)} R_t(-Z)
\end{align}
is multiportfolio time consistent. Furthermore, $\seq{\tilde{R}}$ is $M_t$-translative and satisfies monotonicity, but may fail to be finite
at zero.  Additionally, if $\seq{R}$ is (conditionally) convex ((conditionally) coherent, convex and c.u.c.) then $(\tilde{R}_t)_{t=0}^T$ is  (conditionally) convex ((conditionally) coherent, convex and c.u.c.).
\end{proposition}
\begin{proof}
All but the convex upper continuity property was proven in proposition~3.11 in~\cite{FR12}.

By $\tilde{R}_T = R_T$, then convex upper continuity trivially holds at time $T$. Using backwards induction we will assume $\tilde{R}_{t+1}$ is c.u.c., then $\tilde{R}_t$ is the composition of convex and c.u.c. set-valued functions.  Thus, by proposition~\ref{prop_comp_usc}, $\tilde{R}_t$ is c.u.c.
\qed\end{proof}

\begin{corollary}[Corollary~3.14 in~\cite{FR12}]
\label{cor_composed_1}
Let $\seq{R}$ be a dynamic risk measure on $\LdpF{}$ with acceptance sets $\seq{A}$. Then, the following are equivalent:
\begin{enumerate}
\item \label{cor_composed_risk} $\seq{\tilde{R}}$ is defined as in equations~\eqref{eqn_composed_final} and~\eqref{eqn_composed};
\item \label{cor_composed_acceptance} $\seq{\tilde{A}}$ is defined by
    \begin{align*}
    \tilde{A}_T &= A_T,\\
    \forall t \in \lrcurly{0,1,...,T-1}: \; \tilde{A}_t &= A_{t,t+1} + \tilde{A}_{t+1},
    \end{align*}
    where $\seq{\tilde{A}}$ denotes the acceptance set of $\seq{\tilde{R}}$.
\end{enumerate}
\end{corollary}

\begin{corollary}
\label{cor_composed}
Let the assumptions of corollary~\ref{cor_composed_1} be satisfied.  Additionally let
\begin{enumerate}
\setcounter{enumi}{2}
\item
$\seq{R}$ be c.u.c. and convex with minimal penalty function $\seq{-\beta^{\min}}$.  Then, $\seq{-\tilde{\beta}^{\min}}$ defined recursively by
    \begin{align*}
    -\tilde{\beta}_T^{\min}(\Q_T,w_T) &= -\beta_T^{\min}(\Q_T,w_T),\\
    -\tilde{\beta}_t^{\min}\lrparen{\Q_t,w_t} &= \cl\lrparen{-\beta_{t,t+1}^{\min}\lrparen{\Q_t,w_t} + \EQt{-\tilde{\beta}_{t+1}^{\min}\lrparen{\Q_t^{t+1},w_t^{t+1}\lrparen{\Q_t,w_t}}}{t}}
    \end{align*}
 where $t \in \{0,1,...,T-1\}$ and $(\Q_{t},w_{t}) \in \olW_{t}$, is equivalently defined by
 \begin{equation}
 \label{proofbeta}
 	-\tilde{\beta}^{\min}_t(\Q,w) := \cl\bigcup_{Z \in \tilde{A}_t} \lrparen{\EQt{Z}{t} + G_t(w)} \cap M_t,
 \end{equation}
 where  $\seq{\tilde{A}}$ is obtained by the recursion in property 2 in corollary~\ref{cor_composed_1} .
The dynamic risk measure $\seq{\tilde{R}}$ corresponding to $\seq{\tilde{A}}$  is c.u.c. convex and multiportfolio time consistent (but may fail to be finite at zero).  Further, if $\tilde{R}_t$ is finite at zero then $\tilde{R}_t$ is equivalent to its dual form with penalty function $-\tilde{\beta}^{\min}_t$ and half-spaces $G_t(w)$.
\item
$\seq{R}$ be conditionally c.u.c. and conditionally convex with dual representation \eqref{dual_We} w.r.t. $\olW_t^e$ and minimal penalty function $\seq{-\alpha^{\min}}$.  Then, $\seq{-\tilde{\alpha}^{\min}}$ defined recursively by
    \begin{align*}
    -\tilde{\alpha}_T^{\min}(\Q_T,w_T) &= -\alpha_T^{\min}(\Q_T,w_T),\\
    -\tilde{\alpha}_t^{\min}\lrparen{\Q_t,w_t} &= \cl\lrparen{-\alpha_{t,t+1}^{\min}\lrparen{\Q_t,w_t} + \EQt{-\tilde{\alpha}_{t+1}^{\min}\lrparen{\Q_t^{t+1},w_t^{t+1}\lrparen{\Q_t,w_t}}}{t}}
    \end{align*}
 where $t \in \{0,1,...,T-1\}$ and $(\Q_{t},w_{t}) \in \olW_{t}^e$, is equivalently defined by
 \begin{equation}
 \label{proofalpha}
 	-\tilde{\alpha}^{\min}_t(\Q,w) := \cl\bigcup_{Z \in \tilde{A}_t} \lrparen{\EQt{Z}{t} + \Gamma_t(w)} \cap M_t,
 \end{equation}
 where  $\seq{\tilde{A}}$ is obtained by the recursion in property 2 in corollary~\ref{cor_composed_1} .
The dynamic risk measure $\seq{\tilde{R}}$ corresponding to $\seq{\tilde{A}}$  is conditionally c.u.c., conditionally convex and multiportfolio time consistent (but may fail to be finite at zero).  Further, if $\tilde{R}_t$ is finite at zero then $\tilde{R}_t$ is equivalent to its dual form with penalty function $-\tilde{\alpha}^{\min}_t$ and conditional half-spaces $\Gamma_t(w)$.
\item $\seq{R}$ be (conditionally) c.u.c. and (conditionally) coherent
with maximal dual set $\seq{\olW^{\max}}$. Then,
$\seq{\widetilde{\W}^{\max}}$ defined recursively by
    \begin{align*}
    \widetilde{\W}_{T}^{\max} &= \olW_{T}^{\max},\\
    \widetilde{\W}_{t}^{\max} &= \olW_{t,t+1}^{\max} \cap H_t^{t+1}(\widetilde{\W}_{t+1}^{\max}),
    \end{align*}
    where $t \in \{0,1,...,T-1\}$, is equivalently defined by \[\widetilde{\W}_t^{\max} := \lrcurly{(\Q,w) \in \olW_t: w_t^T(\Q,w) \in \plus{\tilde{A}_t}},\] where  $\seq{\tilde{A}}$ is obtained by the recursion in property 2 in corollary~\ref{cor_composed_1} .
The dynamic risk measure $\seq{\tilde{R}}$ corresponding to $\seq{\tilde{A}}$ is (conditionally) c.u.c., (conditionally) coherent and multiportfolio time consistent, and is finite at zero if and only if
$\widetilde{\W}_{t}^{\max} \neq \emptyset$ for all times $t$.
\end{enumerate}
\end{corollary}

\begin{proof}
\begin{proofenum}
\setcounter{enumi}{2}
\item The proof of corollary~\ref{thm_mptc_penalty} demonstrates the equivalence between the sum of penalty functions and the sum of acceptance sets, where $A$ and $-\beta$ have to be replaced by $\tilde{A}$ and $-\tilde{\beta}$ at the appropriate places. Regarding the assumptions of corollary~\ref{thm_mptc_penalty}:  c.u.c. and convexity follow from proposition~\ref{prop_gen_tc}
and normalization is not needed for this equivalence as stated in remark~5 in \cite{FR12}.  Notice that lemma~\ref{separation} does not require the finite at zero properties for acceptance sets.
 Finally, if $\tilde{R}_t$ is finite at zero, then it is equivalent to its dual representation with minimal penalty function $-\tilde{\beta}^{\min}_t$ by theorem~\ref{thm_probability_equal}.
 \item The proof is analog to 3., using corollary~\ref{cor_mptc_cond_penalty} instead of corollary~\ref{thm_mptc_penalty} and corollary~\ref{cor_conditional_dual} instead of theorem~\ref{thm_probability_equal}. Adapting proposition~\ref{prop_comp_usc} to the conditional case yields $\seq{\tilde{R}}$ conditionally c.u.c.
\item Using the definition in~\eqref{proofbeta}, corollary~\ref{cor_mptc_coherent}, where $\olW,A$ and $-\beta$ is replaced by $\widetilde{\W}, \tilde{A}$ and $-\tilde{\beta}$ at the appropriate places,  yields the equivalence between the two definitions of $\seq{\widetilde{\W}^{\max}}$. Convex upper continuity and (conditionally) coherence follow from proposition~\ref{prop_gen_tc}.
Conditional convex upper continuity follows if $\seq{R}$ is conditionally coherent by adapting proposition~\ref{prop_comp_usc} to the conditional case.
Additionally, $\tilde{R}_t(0) \neq \emptyset$, see proof of proposition~3.12 in \cite{FR12}. Furthermore, $\tilde{R}_t(0) \neq M_t$ implies that $\tilde{R}_t$ is proper and thus the dual representation holds true. Then, there exists a $(\Q,w) \in \olW_t$ such that $-\tilde{\beta}_t^{\min}(\Q,w) \neq M_t$, i.e. $\widetilde{\W}_t^{\max} \neq \emptyset$. And $\tilde{R}_t(0) = M_t$ implies, by proposition~13~(iv) in~\cite{H09}, $M_t = \tilde{R}_t(0) \subseteq\tilde{R}^{**}_t(0) = \bigcap_{(\Q,w) \in \widetilde{\W}_t^{\max}} G_t(w) \cap M_t$ and thus $\widetilde{\W}_t^{\max} = \emptyset$.
\end{proofenum}
\qed\end{proof}

We will now use the above results to show that the convex superhedging portfolios are multiportfolio time consistent, and to deduce a multiportfolio time consistent version of the average value at risk by backward recursion.
\begin{example}[Convex superhedging price]
\label{ex_convex_shp}
Consider the setting with a full space of eligible assets, i.e. $M_t = \LdpF{t}$ for all times $t$.  Also consider a market with convex transaction costs as in \cite{PP10}, which is modeled by a sequence of convex solvency regions $\seq{K}$. $K_t$ is a solvency region at time $t$ if it is an $\Ft{t}$-measurable set such that for every $\omega \in \Omega$, $K_t[\omega]$ is a closed set with $\mathbb{R}_+^d \subseteq K_t[\omega] \subsetneq \mathbb{R}^d$. Let the appropriate no arbitrage condition (robust no scalable arbitrage, see  \cite{PP10}) be satisfied.

Denote the set of self-financing portfolios starting from zero capital at time $t$ by $C_{t,T} := -\sum_{s = t}^T \LdpK{p}{s}{K_s}$.
Thus the set of superhedging portfolios is given by
\[CSHP_t(X) := \lrcurly{u \in \LdpF{t}: -X + u \in -C_{t,T}}.\]
The convex superhedging portfolios can also be defined via the dual representation with penalty functions
\begin{align*}
-\alpha_t^{CSHP}(\Q,w) &:= \sum_{s = t}^T \lrcurly{u \in \LdpF{t}:
    \essinf_{k \in \LdpK{p}{s}{K_s}} \trans{w}\EQt{k}{t} \leq \trans{w}u},\\
-\beta_t^{CSHP}(\Q,w) &:= \sum_{s = t}^T \lrcurly{u \in \LdpF{t}: \sigma_{K_s}(w_t^s(\Q,w)) \leq \E{\trans{w}u}}
\end{align*}
where $\sigma_{K_s}(w_t^s(\Q,w)) = \inf_{k \in \LdpK{p}{s}{K_s}} \inE{\trans{w_t^s(\Q,w)}k}$ is the support function for
the selectors of $K_s$.

One can use corollary~\ref{cor_composed_1} and proposition~\ref{prop_gen_tc} to show that the convex superhedging portfolios are
multiportfolio time consistent: Consider acceptance sets $\seq{A}$ given by $A_T = \LdpK{p}{T}{K_T}$ and $A_t=\LdpK{p}{t}{K_t}
+ \LdpF{+}$ for $t<T$. Thus, the stepped acceptance sets are given by $A_{t,t+1} = \LdpK{p}{t}{K_t} + \LdpF{t+1,+}$. Then, the acceptance set $-C_{t,T}$ of the convex superhedging set can be recovered by backward recursion of $\seq{A}$, that is $-C_{T,T} =A_T$ and $-C_{t,T} = -C_{t+1,T} + A_{t,t+1}$ for $t<T$. Thus, by corollary~\ref{cor_composed_1} and proposition~\ref{prop_gen_tc} the convex superhedging portfolios are multiportfolio time consistent.

Under the robust no scalable arbitrage condition $-C_{t,T}$ is closed.  Therefore, by theorem~\ref{thm_mptc_penalty_wo_uc}, convex upper continuity is not necessary in corollary~\ref{cor_composed}, or for the cocycle condition to be satisfied.  And indeed, we can recover the minimal penalty function $-\beta_t^{CSHP}$ by the backward recursion of penalty functions as given in corollary~\ref{cor_composed}
\begin{align*}
-\beta_T^{CSHP}(\Q,w) &= -b_T(\Q,w)\\
-\beta_t^{CSHP}(\Q,w) &= \cl\lrparen{-b_{t,t+1}(\Q,w) + \EQt{-\beta_{t+1}^{CSHP}(\Q^{t+1},w_t^{t+1}(\Q,w))}{t}}
\end{align*}
for any $(\Q,w) \in \olW_t$, where
\begin{align*}
-b_T(\Q,w) &:= \cl \bigcup_{X \in A_{T}} \lrparen{\EQt{X}{t} + G_t(w)}= \lrcurly{u \in \LdpF{}: \sigma_{K_T}(w) \leq \E{\trans{w}u}},\\
-b_{t,t+1}(\Q,w) &:= \cl \bigcup_{X \in A_{t,t+1}} \lrparen{\EQt{X}{t} + G_t(w)}= \lrcurly{u \in \LdpF{t}: \sigma_{K_t}(w) \leq \E{\trans{w}u}}.
\end{align*}
\end{example}

\begin{example}[Composed $AV@R$]
\label{ex_avar}
The set-valued average value at risk was shown not to be multiportfolio time consistent in~\cite{FR12} (and similarly the scalar average value at risk is well known to not be time consistent).  Corollary~\ref{cor_composed} can be used to construct the composed version of the average value at risk and deduce its dual representation.

Consider $p = +\infty$ with the weak* topology and parameters $\lambda^t \in \LdiF{t}$ with bounds $\epsilon \leq \lambda^t_i < 1$ for some $\epsilon > 0$ for every time $t$.  The details for the average value at risk in this setting are provided in section~\ref{sec_avar}.

Let $M_t = \LdiF{t}$ for all times $t$, then as shown in proposition~\ref{prop_avar_usc} we know that $\seq{AV@R^{\lambda}}$ is a c.u.c. dynamic conditionally coherent risk measure.
Then, the composed version of the average value at risk $\seq{AV@R^{\lambda}}$ is, by corollary~\ref{cor_composed}, a multiportfolio time consistent c.u.c.  conditionally coherent risk measure with dual representation
\begin{equation*}
\widetilde{AV@R}_t^{\lambda}(X) := \bigcap_{(\Q,w) \in \widetilde{\W}^{\lambda}_{t}} \lrparen{\EQt{-X}{t} + \Gamma_t(w)} \cap M_t,
\end{equation*}
where
\begin{align*}
\widetilde{\W}^{\lambda}_{t}
&= \Big\{(\Q,w) \in \olW_t: \forall s \in \lrcurly{t,...,T-1}, \frac{w_{t}^{s}(\Q,w)}{\lambda^{s}} \succeq w_{t}^{s+1}(\Q,w) \Big\}\\
&= \Big\{(\Q,w) \in \olW_t: \forall s \in \lrcurly{t,...,T-1}, \P\lrparen{\bar{\xi}_{s,s+1}(\Q_i) \leq \frac{1}{\lambda^{s}_i} \text{ or } w_i = 0} = 1 \; \forall i \in \{1,...,d\} \Big\}.
\end{align*}
\end{example}

%%%%%%%%%%%%%%%%%%%%%%%%%%%%%
\section{Detailed examples}

\subsection{Average Value at Risk}
\label{sec_avar}
In this section we will discuss the details for the dynamic set-valued average value at risk and prove the dual representation of the composed
dynamic set-valued average value at risk  given in example~\ref{ex_avar} by using corollary~\ref{cor_composed}.
In the scalar case the composed average value at risk is studied in \cite{CK10}.
As the underlying space we consider  $\LdiF{t}$ with the weak* topology $\sigma(\LdiF{t},\LdoF{t})$.

The dual definition for the dynamic average value at risk with time $t$ parameters $\lambda^t \in \LdiF{t}$ with $\epsilon \leq \lambda^t_i < 1$ for some $\epsilon > 0$ is given by
\begin{equation}
\label{AVARdualdef}
AV@R_t^{\lambda}(X) := \bigcap_{(\Q,w) \in \olW^{\lambda}_{t}} \lrparen{\EQt{-X}{t} + \Gamma_t(w)} \cap M_t
\end{equation}
for any $X \in \LdiF{}$ where
\begin{align*}
\olW^{\lambda}_{t} &:= \lrcurly{(\Q,w) \in \olW_t: \frac{w}{\lambda^t} - w_t^T(\Q,w) \in \LdoF{+}}= \lrcurly{(\Q,w) \in \olW_t: 0 \preceq \diag{w}\dQdP \preceq w/\lambda^t},
\end{align*}
see section~5.2 in \cite{FR12}.

In the following proposition, we provide the acceptance set and thus the primal representation for the dynamic average value at risk given in \eqref{AVARdualdef}. This proves that \eqref{AVARdualdef} is the dynamic version of the closure of the static average value at risk defined via its acceptance set $A_0^{\lambda}$ in \cite{HRY12}, as the proof is similar to \cite{HRY12} we choose to omit it.
\begin{proposition}
\label{prop_avar_acceptance}
The acceptance set associated with the conditional average value at risk at time $t$ and parameter $\lambda^t$ is given by $\bar{A}_t^{\lambda} = \cl (A_t^{\lambda})$ where
\[A_t^{\lambda} = \lrcurly{X \in \LdiF{}: \exists Z \in \LdiF{+}, X + Z \succeq \frac{\Et{Z}{t}}{\lambda^t}}\]
and $\olW^{\lambda}_t$ is the maximal dual set.
\end{proposition}

In proposition~5.4 in~\cite{FR12} it was shown that $\seq{AV@R^{\lambda}}$ is a normalized closed conditionally coherent dynamic risk measure.  And
in proposition~\ref{prop_avar_usc} below, we show that the average value at risk with $M_t = \LdiF{t}$ is c.u.c.

\begin{proposition}
\label{prop_avar_usc}
Let $M_t = \LdiF{t}$ for all times $t$, then $\seq{AV@R^{\lambda}}$ is a c.u.c. risk measure.
\end{proposition}
\begin{proof}
Let $X \in \LdiF{}$, then
\begin{align*}
AV@R_t^{\lambda}(X) &= \lrcurly{u \in \LdiF{t}: X + u \in \bar{A}_t^{\lambda}}= \lrcurly{u \in \LdiF{t}: \forall i = 1,...,d: X_i + u_i \in \bar{A}_t^{\lambda,i}}
\end{align*}
where \[\bar{A}_t^{\lambda,i} = \cl\lrcurly{X \in \LiF{}: \exists Z \in L^\infty(\R_{+}), X + Z \geq \frac{1}{\lambda_i^t} \Et{Z}{t}}.\]
Therefore, $u \in AV@R_t^{\lambda}(X)$ if and only if  $u \in \LdiF{t}$ with $u_i \geq \rho_t^{\lambda_i}(X_i)$ $\P$-almost surely, where $\rho_t^{\lambda_i}$ is the scalar dynamic
average value at risk.  The result then follows by proposition~\ref{prop_uc_scalar}.
\qed\end{proof}

We conclude the discussion of the average value at risk by considering the stepped version.

\begin{lemma}
\label{lemma_avar_stepped}
The stepped average value at risk from time $t$ to $s$ (for $0 \leq t < s \leq T$) with time $t$ parameters $\lambda^t \in \LdiF{t}$ where $\epsilon \leq \lambda^t_i < 1$ for some $\epsilon > 0$ is given by
\begin{equation*}
AV@R_{t,s}^{\lambda}(X) := \bigcap_{(\Q,w) \in \olW^{\lambda}_{t,s}} \lrparen{\EQt{-X}{t} + \Gamma_t(w)} \cap M_t
\end{equation*}
for any $X \in \LdiF{}$ where
\begin{align*}
\olW_{t,s}^{\lambda} &= \lcurly{(\Q,w) \in \olW_{t,s}: \forall Z \in \LdiF{s,+}, \; \E{\transp{w/\lambda^t - w_t^{s}(\Q,w)} Z} \geq}\\
&\quad\quad \sup\Big\{\E{\trans{w_t^{s}(\Q,w)} D}: \rcurly{D \in \LdiF{s,-} \cap \lrsquare{M_{s} + \lrparen{\Et{Z}{t}/\lambda^t - Z}}\Big\}}
\end{align*}
is the associated maximal stepped dual set.
\end{lemma}

\begin{proof}
Using the definition of the acceptance set for $AV@R_t^{\lambda}$ given in proposition~\ref{prop_avar_acceptance}, we find the stepped acceptance set is given by $\bar{A}_{t,s}^{\lambda} = \cl (A_{t,s}^{\lambda})$ where
\begin{align*}
A_{t,s}^{\lambda} &= \lrcurly{X \in M_{s}: \exists Z \in \LdiF{+}, X + Z \succeq \Et{Z}{t}/\lambda^t}=  \lrcurly{X \in M_{s}: \exists Z \in \LdiF{s,+}, X + Z \succeq \Et{Z}{t}/\lambda^t}\\
&= \lrparen{\bigcup_{Z \in \LdiF{s,+}} \lrparen{\Et{Z}{t}/\lambda^t - Z} + \LdiF{s,+}} \cap M_{s}.
\end{align*}
By corollary~\ref{cor_stepped_rep} and $\plusp{A_{t,s}^{\lambda}} = \plusp{\bar{A}_{t,s}^{\lambda}}$, the maximal stepped dual set is given by \[\lrcurly{(\Q,w) \in \olW_{t,s}: w_t^{s}(\Q,w) \in \plusp{A_{t,s}^{\lambda}}}.\]

It can be seen that $X \in A_{t,s}^{\lambda}$ if and only if $X = \inEt{Z}{t}/\lambda^t - Z + D$ for some $Z \in \LdiF{s,+}$
and $D \in \LdiF{s,+} \cap \inlrsquare{M_{s} + (Z - \inEt{Z}{t}/\lambda^t)}$.  Therefore $w_t^{s}(\Q,w) \in \plusp{A_{t,s}^{\lambda}}$ if
and only if for any $Z \in \LdiF{s,+}$ and $D \in \LdiF{s,+} \cap \inlrsquare{M_{s} + (Z - \inEt{Z}{t}/\lambda^t)}$
\begin{align*}
0 &\leq \E{\trans{w_t^{s}(\Q,w)} \lrparen{\Et{Z}{t}/\lambda^t - Z + D}}= \E{\transp{w/\lambda^t - w_t^{s}(\Q,w)} Z} + \E{\trans{w_t^{s}(\Q,w)} D}.
\end{align*}
That is, for every $Z \in \LdiF{s,+}$
\begin{align*}
&\sup\lrcurly{\E{\trans{w_t^{s}(\Q,w)} D}: D \in \LdiF{s,-} \cap \lrsquare{M_{s} + \lrparen{\Et{Z}{t}/\lambda^t - Z}}}\leq \E{\transp{w/\lambda^t - w_t^{s}(\Q,w)} Z}.
\end{align*}
\qed\end{proof}

\begin{remark}
\label{rem_avar_stepped_full}
The dual representation in lemma~\ref{lemma_avar_stepped} simplifies significantly if all assets are eligible, i.e., if $M_t = \LdiF{t}$ for all times $t$. Then, the maximal dual sets for the stepped average value at risk can be equivalently given by
\[\olW_{t,s}^{\lambda} = \lrcurly{(\Q,w) \in \olW_t: 0 \preceq \diag{w}\xi_{t,s}(\Q) \preceq w/\lambda^t}\]
for all times $0 \leq t < s \leq T$, where $\olW_{t,s} = \olW_t$ by remark~\ref{rem_stepped_dual_var_equal}.  This dual representation can be interpreted as the extension of the stepped scalar representation given in~\cite{CK10}.
\end{remark}

We will now prove the dual representation of the composed, multiportfolio time consistent version of $\seq{AV@R^{\lambda}}$ given in example~\ref{ex_avar}.  As in section~\ref{sec_composition} for composed risk measures, we will now work in the discrete time setting $t \in \{0,1,...,T\}$.

\begin{proof}[Example~\ref{ex_avar}]
By corollary~\ref{cor_composed}, $\seq{\widetilde{AV@R}^{\lambda}}$ is the multiportfolio time consistent version of $\seq{AV@R^{\lambda}}$ if and only if
\begin{align*}
\widetilde{\W}^{\lambda}_{T} & = \olW^{\lambda}_{T}\\
\widetilde{\W}^{\lambda}_{t} & = H_t^{t+1}\lrparen{\widetilde{\W}^{\lambda}_{t+1}} \cap \olW^{\lambda}_{t,t+1}
\end{align*}
where $\widetilde{\W}^{\lambda}_t \neq \emptyset$ for all times $t$.  Trivially it can be seen that
$\widetilde{\W}^{\lambda}_{T} = \olW_T$.  Furthermore, $\olW^{\lambda}_{T} = \olW_T$ since $\frac{1}{\lambda^T_i} - 1 \geq 0$
for every $i = 1,...,d$ (by $\epsilon \leq \lambda^T_i < 1$) and $w_T^T(\Q,w) \in \LdoF{+}$ (by $(\Q,w) \in \olW_T$ and by
noting $w_T^T(\Q,w) = w$), and therefore the product is almost surely nonnegative.
By remark~\ref{rem_avar_stepped_full} and lemma~\ref{lemma_avar_stepped} it holds
\[\olW_{t,s}^{\lambda} = \lrcurly{(\Q,w) \in \olW_t: w/\lambda^t \succeq w_t^{s}(\Q,w)}\]
Furthermore, using lemma~\ref{lemma_Wtau_to_Wt} and $w_{t+1}^{s}(\Q^{t+1},w_t^{t+1}(\Q,w))=w_t^{s}(\Q,w)$ it follows
\begin{align*}
H_t^{t+1}\lrparen{\widetilde{\W}^{\lambda}_{t+1}} = &\Big\{(\Q,w) \in \olW_t: \forall s \in \lrcurly{t+1,...,T-1}, \frac{w_{t}^{s}(\Q,w)}{\lambda^{s}} \succeq w_t^{s+1}(\Q,w)\Big\}.
\end{align*}
Since $w_{s}^{s}(\Q,w) = w$ for any time $s$, the recursive form $\widetilde{\W}^{\lambda}_{t}  = H_t^{t+1}\inlrparen{\widetilde{\W}^{\lambda}_{t+1}} \cap \olW^{\lambda}_{t,t+1}$ is proven.
$\widetilde{\W}_t^{\lambda} \neq \emptyset$ holds since $(\P,w) \in \widetilde{\W}_t^{\lambda}$ for any $w \in \LdoF{t,+} \backslash \{0\}$.  This is because $(\P,w) \in \olW_t$ and for any $s \in \{t,...,T\}$ it follows $w_t^{s}(\P,w) = w$ and
\[\E{\transp{w/\lambda^{s} - w} Z} \geq 0\]
for every $Z \in \LdiF{+}$.

Finally, $w_{t}^{s}(\Q,w)/\lambda^{s} - w_{t}^{s+1}(\Q,w) \in \LdoF{s+1,+}$ if and only if it is componentwise nonnegative, i.e.
$w_{t}^{s}(\Q,w)_i \inlrparen{\frac{1}{\lambda^{s}_i} - \bar{\xi}_{s,s+1}(\Q_i)} \geq 0$ almost surely for every $i = 1,...,d$.
Since $(\Q,w) \in \olW_t$ we know $w_{t}^{s}(\Q,w)_i \geq 0$, therefore $(\Q,w) \in \widetilde{\W}_t^{\lambda}$ if and only if $\P\inlrparen{\bar{\xi}_{s,s+1}(\Q_i) \leq \frac{1}{\lambda^{s}_i} \text{ or } w_{t}^{s}(\Q,w)_i = 0} = 1$ for every $i \in \{1,...,d\}$.  Notice that for any $\omega \in \Omega$ we have $w_t^{s}(\Q,w)_i[\omega] = 0$ if and only if $w_i[\omega] = 0$ or $\bar{\xi}_{t,r}(\Q_i)[\omega] = 0$ for some time $r \in (t,s]$, but $\bar{\xi}_{t,r}(\Q_i)[\omega] = 0$ implies $\bar{\xi}_{s,s+1}(\Q_i)[\omega] = 1 \leq \frac{1}{\lambda^{s}_i[\omega]}$.  Thus we recover the final form
\begin{align*}
\widetilde{\W}_t^{\lambda} &= \Big\{(\Q,w) \in \olW_t: \forall s \in \lrcurly{t,...,T-1}, \P\lrparen{\bar{\xi}_{s,s+1}(\Q_i) \leq \frac{1}{\lambda^{s}_i} \text{ or } w_i = 0} = 1 \; \forall i \in \{1,...,d\}\Big\}.
\end{align*}
\qed\end{proof}

\subsection{Entropic risk measure}
\label{sec_entropic}

The set-valued entropic risk measure was studied in~\cite{AHR13} in a single period static framework.  We will present a
dynamic version of the entropic risk measure.  In
example~\ref{ex_entropic}, we presented and worked with the restrictive entropic risk measure only.

For the purposes of this section let $p
= +\infty$ and $q = 1$, and consider the weak* topology. Let $M_t = \LdiF{t}$ for all times $t$. Further, consider parameters $\lambda^t \in \LdiF{t,++}$ with
$\lambda^t_i \geq \epsilon$ for some $\epsilon > 0$ for every index $i$, and let $C_t \in
\mathcal{G}(\LdiF{t};\LdiF{t,+})$ with $0 \in C_t$ and $C_t \cap \LdiF{t,--} = \emptyset$. The set $C_t$ will model the set of acceptable expected utilities, thus $C_t = \LdiF{t,+}$ is the most restrictive (conservative) choice.

The dynamic entropic risk measure with parameters $\lambda^t$ and $C_t$ can then be defined by
\begin{equation*}
R_t^{ent}(X;\lambda^t,C_t) := \lrcurly{u \in \LdiF{t}: \Et{u_t(X + u)}{t} \in C_t }
\end{equation*}
$u_t(x) =\transp{u_{t,1}(x_1),...,u_{t,d}(x_d)}$ for any $x \in \R^d$ and $u_{t,i}(z)=\frac{1-e^{-\lambda^t_i z}}{\lambda^t_i}$ for $z\in\R$ and $i=1,...,d$.

It can be shown
in analogy to propositions~4.4 and~5.1 of~\cite{AHR13} that the entropic risk measure is conditionally convex, closed and equal to
\begin{equation}\label{entpoint}
	R_t^{ent}(X;\lambda^t,C_t) = \rho_t^{ent}(X;\lambda^t) + \tilde{C}_t(\lambda^t,C_t)
	\end{equation}
for any $X \in \LdiF{}$, where
\[\rho_t^{ent}(X;\lambda^t) = \frac{\log\lrparen{\Et{\exp(-\diag{\lambda^t}X)}{t}}}{\lambda^t}\]
and
\[\tilde{C}_t(\lambda^t,C_t) = -\frac{\log\lrsquare{\lrparen{\1 - \diag{\lambda^t}C_t} \cap \LdiF{t,++}}}{\lambda^t},\]
with $\1 = \transp{1,...,1} \in \R^d$ and the exponential and logarithm are taken componentwise for a vector and elementwise for a set, e.g. $\exp(z) =
\transp{\exp(z_1),...,\exp(z_d)}$ for any $z \in \R^d$.

The dual form of the dynamic entropic risk measure can be deduced as follows.  This is a trivial extension from the work in~\cite{AHR13}, so we will omit the proof in this paper.
\begin{lemma}
\label{lemma_entropic_dual}
Let $0 \leq t < s \leq T$.  The minimal stepped penalty functions of the stepped entropic risk measure are given by
\begin{align*}
-\alpha_{t,s}^{ent}(\Q,w;\lambda^t,C_t) & := -\frac{\hat{H}_{t,s}(\Q|\P)}{\lambda^t} + \tilde{C}_t(\lambda^t,C_t) +
\Gamma_t(w),\\
-\beta_{t,s}^{ent}(\Q,w;\lambda^t,C_t) &:= -\frac{\hat{H}_{t,s}(\Q|\P)}{\lambda^t} + \tilde{C}_t(\lambda^t,C_t) + G_t(w)
\end{align*}
for any $(\Q,w) \in \olW_t$ where \[\hat{H}_{t,s}(\Q|\P) := \EQt{\log(\xi_{t,s}(\Q))}{t}.\]

The dual representation of the entropic risk measure is given by \eqref{convex_dual} with minimal penalty function $-\beta_t^{ent} := -\beta_{t,T}^{ent}$. Note that
$\hat{H}_{t,T}(\Q|\P) = \inEQt{\log(\dQdP)}{t}$ is the conditional relative entropy.
\end{lemma}

Let $C_t = \LdiF{t,+}$ almost surely for all times $t$ and consider a constant risk aversion level $\lambda \in \R^d_{++}$.
It can be seen that $\tilde{C}_t(\lambda,\LdiF{t,+}) = \LdiF{t,+}$, therefore it immediately follows that
$(R^{ent}_t(\cdot;\lambda,C_t))_{t = 0}^T$ is normalized.  Therefore the results from example~\ref{ex_entropic} all follow
trivially.

%%%%%%%%%%%%%%%%%%%%%%%%%%%%%%%%%%%%%%%%
%%%%%%%%%%%%%%%%%%%%%%%%%%%%%%%%%%%%%%%%
\appendix\normalsize
\label{sec_appendix}

%%%%%%%%%%%%%%%%%%%%%%%%%%%%%%%%%%%%%%%%%%%%%%%%%%%%
\section{On the relationship of dual variables at different times}
\label{sec_dualvar}
In considering how closed convex (and coherent) risk measures relate through time we must consider how the sets of dual variables
relate.
In the following lemma we provide such a relationship between elements of $\olW_t$ and elements of $\olW_{s}$ for any times $t,s$ with
$t \leq s$.  In fact we define a mapping on $\olW_t$ which is equivalent (in the set-valued replacement for continuous linear functionals) in
$\olW_{s}$.  In the scalar framework this type of property is not needed since the set $\olW_t$ can be simplified to any $\Q \ll \P$ for any time $t$ (where $\Q = \P|_{\Ft{t}}$).

\begin{lemma}
\label{lemma_Wtau_to_Wt}
For any choice of times $t$ and $s>t$ it follows that:
\begin{enumerate}
\item $\inlrcurly{(\Q^{s},w_t^{s}(\Q,w)): (\Q,w) \in \olW_t} \subseteq \olW_{s}$,
\item for every $(\R,v) \in \olW_{s}$ there exists $(\Q,w) \in \olW_t$ such that $F_{(\R,v)}^s = F_{(\Q^{s},w_t^{s}(\Q,w))}^s$.
\end{enumerate}
\end{lemma}

\begin{proof}
\begin{proofenum}
\item $\inlrcurly{(\Q^{s},w_t^{s}(\Q,w)): (\Q,w) \in \olW_t} \subseteq \olW_{s}$ if and only if for every pair $(\Q,w) \in \olW_t$ it follows that $w_t^{s}(\Q,w) \in \plus{M_{s,+}} \backslash \prp{M_{s}}$ and $w_{s}^T(\Q^{s},w_t^{s}(\Q,w)) \in \LdqF{+}$.
    \begin{proofenum}
    \item Let $(\Q,w) \in \olW_t$. Show $w_t^{s}(\Q,w) \in \plus{M_{s,+}} \backslash \prp{M_{s}}$:
        \begin{proofenum}
        \item Let $m_{s} \in M_{s,+}$, then \[\E{\trans{w_t^{s}(\Q,w)} m_{s}} = \E{\trans{w} \EQt{m_{s}}{t}} \geq 0\] since $\inEQt{m_{s}}{t} \in M_{t,+}$ by $M_t \supseteq M_{s} \cap \LdpF{t}$ and $M_s = \LdpK{p}{s}{M}$.
        \item Since  $(\Q,w) \in \olW_t$, in particular since $w\notin\prp{M_{t}}$ there exists $m_t \in M_t\subseteq M_{s}$ such that
$\inE{\trans{w} m_t} \neq 0$. Then, \[\E{\trans{w_t^{s}(\Q,w)} m_t} = \E{\trans{w} \EQt{m_t}{t}} = \E{\trans{w} m_t} \neq 0.\]
        \end{proofenum}
    \item $w_{s}^T(\Q^{s},w_t^{s}(\Q,w)) = w_t^T(\Q,w) \in \LdqF{+}$ by $(\Q,w) \in \olW_t$.
    \end{proofenum}

\item By lemma 4.5 in~\cite{FR12}, for every $(\R,v) \in \olW_{s}$ there exists a $(Y,\bar{v})$ such that $Y
\in \LdqF{+}$, $\bar{v} \in \inlrparen{\inEt{Y}{s} + \prp{M_{s}}} \backslash \prp{M_{s}}$ such that
$F_{(\R,v)}^s = \tilde{F}_{(Y,\bar{v})}^s$.  And for every $(Y,\bar{v})$ with $Y \in
\LdqF{+}$ and $\bar{v} \in \inlrparen{\inEt{Y}{s} + \prp{M_{s}}} \backslash \prp{M_{s}}$ there
exists $(\hat{\Q},w_{s})\in \olW_{s}$ such that $\tilde{F}_{(Y,\bar{v})}^{s} =
F_{(\hat{\Q},w_{s})}^{s}$ by setting $w_{s} = \inEt{Y}{s}$ and
\[\bar{\xi}_{r,s}^i[\omega] = \begin{cases}\frac{\Et{Y_i}{s}(\omega)}{\Et{Y_i}{r}(\omega)} & \text{if
} \Et{Y_i}{r}(\omega) > 0\\ 1 & \text{else} \end{cases}\]
for every $\omega \in \Omega$, and $\frac{d\hat{\Q}_i}{d\P} = \bar{\xi}_{s,T}^i$.  Define $\Q \in \mathcal{M}$ by $\dQidP = \bar{\xi}_{t,T}^i$, thus $\Q^{s} = \hat{\Q}$.
    Therefore it remains to show that there exists a $w_t \in \LdqF{t}$ such that $w_{s} = w_t^{s}(\Q,w_t)$ and $(\Q,w_t) \in \olW_t$.
Let $w_t := \inEt{w_{s}}{t}=\inEt{Y}{t}$.
    \begin{proofenum}
    \item Show $w_{s} = w_t^{s}(\Q,w_t)$, i.e.
        show $(w_{s})_i(\omega) = (w_t)_i(\omega) \bar{\xi}_{t,s}^i(\omega)$ for every index $i = 1,...,d$ and almost every $\omega \in
\Omega$.  We know if $\inEt{Y_i}{t}(\omega) = 0$ then $(w_t)_i(\omega) = 0$ and $(w_{s})_i(\omega) = 0$ and thus $(w_{s})_i(\omega) = (w_t^{s}(\Q,w_t))_i[\omega]$.
        If $\inEt{Y_i}{t}(\omega) > 0$ then
        \begin{align*}
        w_t^{s}(\Q,w_t)_i[\omega] &= \Et{Y_i}{t}(\omega) \frac{\Et{Y_i}{s}(\omega)}{\Et{Y_i}{t}(\omega)}
        = \Et{Y_i}{s}(\omega) = (w_{s})_i(\omega).
        \end{align*}
    \item Show $(\Q,w_t) \in \olW_t$
        \begin{proofenum}
        \item Show $w_t \in \plus{M_{t,+}} \backslash \prp{M_t}$.
            \begin{proofenum}
            \item Let $m_t \in M_{t,+}$, then $\inE{\trans{w_t} m_t} = \inE{\trans{\inEt{w_{s}}{t}} m_t} = \inE{\trans{w_{s}} m_t} \geq 0$ by the tower property, $M_{t,+} \subseteq M_{s,+}$ and $w_{s}\in \plus{M_{s,+}} $.
            \item Since  $(\Q^{s},w_{s}) \in \olW_{s}$, in particular since $w_{s}\notin\prp{M_{s}}$ there exists $m_{s} \in M_{s}$
such that $\inE{\trans{w_{s}} m_{s}} \neq 0$.  Then $\inEQt{m_{s}}{t} \in M_t$ by $M_t \supseteq M_{s} \cap \LdpF{t}$ and $M_s = \LdpK{p}{s}{M}$.             Therefore, $\inE{\trans{w_t} \inEQt{m_{s}}{t}} = \inE{\trans{w_t^{s}(\Q,w_t)} m_{s}} = \inE{\trans{w_{s}} m_{s}} \neq 0$.
            \end{proofenum}
        \item $w_t^{T}(\Q,w_t) = w_{s}^T(\Q^{s},w_t^{s}(\Q,w_t)) = w_{s}^T(\hat{\Q},w_{s}) \in \LdqF{+}$.
        \end{proofenum}
    \end{proofenum}
\end{proofenum}
\qed\end{proof}

The following corollary of lemma~\ref{lemma_Wtau_to_Wt} uses the above result applied to penalty functions instead of the functionals $ F_{(\cdot,\cdot)}[\cdot]$.

\begin{corollary}
\label{cor_Wtau_to_Wt_penalty}
For any $(\R,v) \in \olW_{s}$ there exists $(\Q,w) \in \olW_t$ such that
\[
-\beta_{s}^{\min}(\R,v) = -\beta_{s}^{\min}(\Q^{s},w_t^{s}(\Q,w))
\]
for any times $0 \leq t < s \leq T$.
\end{corollary}
\begin{proof}
\begin{align*}
 -\beta_{s}^{\min}(\R,v) = \cl \bigcup_{Z \in A_{s}} F_{(\R,v)}^{s}[Z]
  = \cl \bigcup_{Z \in A_{s}} F_{(\Q^{s},w_t^{s}(\Q,w))}^{s}[Z] = -\beta_{s}^{\min}(\Q^{s},w_t^{s}(\Q,w)),
\end{align*}
where the second equation is a result of lemma~\ref{lemma_Wtau_to_Wt}.
\qed\end{proof}

Lemma~\ref{lemma_Wtau_to_Wt} and corollary~\ref{cor_Wtau_to_Wt_penalty} show that for any times $t \leq s$ and for a given penalty
function $-\beta_{s}^{\min}$ the set of dual variables $\inlrcurly{(\Q^{s},w_t^{s}(\Q,w)): (\Q,w) \in \olW_t}$ defines the same closed and convex risk measure at time $s$  as the set of dual variables $\olW_{s}$, that is
\begin{align*}
R_{s}(X)& = \bigcap_{(\Q,w) \in \olW_t} \lrsquare{-\beta_{s}^{\min}(\Q^{s},w_t^{s}(\Q,w)) + \lrparen{\EQt{-X}{s} +
G_{s}\lrparen{w_t^{s}(\Q,w)}}} \cap M_{s}.
\end{align*}

The following lemma, about the expectation of minimal penalty functions, is an extension of lemma~2.6 in~\cite{FP06}.  As a set-valued
operation, this theorem gives a set-valued version of when the conditional expectation of an infimum is equivalent to the infimum of
the conditional expectation.
The proof of the lemma is a simplified version of the proof of lemma 2.6 in~\cite{FP06} since the sets $\{\inEQt{X}{t} + G_t(w)\}$ are
shifted half-spaces for any $X \in A_{t}$ and a fixed $(\Q,w) \in \olW_t$ and thus are completely ordered, in contrast to the scalar
case, where the points $\inEQt{X}{t} $ under consideration are not completely ordered.

\begin{lemma}
\label{lemma_exp_penalty}
For any times $0 \leq t < s \leq T$, and if $R_t$ is a closed convex risk measure, then for any $(\Q,w) \in \olW_t$, it follows that
\[\EQt{-\beta_{s}^{\min}(\Q^{s},w_t^{s}(\Q,w))}{t} = \cl \bigcup_{X \in A_{s}} \lrparen{\EQt{X}{t} + G_t(w)} \cap M_t.\]
\end{lemma}

\begin{proof}
Let $(\Q,w) \in \olW_t$. Then, by lemma~\ref{lemma_Wtau_to_Wt}, $(\Q^{s},w_t^{s}(\Q,w))\in \olW_{s}$. It holds
\begin{align*}
-\beta_{s}^{\min}(\Q^{s},w_t^{s}(\Q,w)) &= \cl \bigcup_{X \in A_{s}} \lrparen{\EQt{X}{s} + G_{s}(w_t^{s}(\Q,w))} \cap M_{s}\\
&= \cl \bigcup_{X \in A_{s}} \lrcurly{u \in M_{s}: \E{\trans{w_t^{s}(\Q,w)}\EQt{X}{s}} \leq \E{\trans{w_t^{s}(\Q,w)}u}}\\
&= \cl \bigcup_{X \in A_{s}} \lrcurly{u \in M_{s}: \E{\trans{w}\EQt{X}{t}} \leq \E{\trans{w}\EQt{u}{t}}}\\
&= \lrcurly{u \in M_{s}: \inf_{X \in A_{s}} \E{\trans{w}\EQt{X}{t}} \leq \E{\trans{w}\EQt{u}{t}}}.
\end{align*}
Taking the conditional expectation on both sides yields
\begin{align*}
\EQt{-\beta_{s}^{\min}(\Q^{s},w_t^{s}(\Q,w))}{t}
&= \Big\{\EQt{u}{t}: u \in M_{s}, \inf_{X \in A_{s}} \E{\trans{w}\EQt{X}{t}} \leq \E{\trans{w}\EQt{u}{t}}\Big\}\\
&= \lrcurly{u \in M_t:\inf_{X \in A_{s}} \E{\trans{w}\EQt{X}{t}} \leq  \E{\trans{w}u}}\\
&= \cl \bigcup_{X \in A_{s}} \lrparen{\EQt{X}{t} + G_t(w)} \cap M_t.
\end{align*}
\qed\end{proof}

One can now show that the $\Q$-conditional expectation (at time $t$) of the positive half-space defined by
$w_t^{s}(\Q,w)$ is given by the positive half-space defined by $w$.

\begin{corollary}
\label{prop_Gw}
Let $0 \leq t < s \leq T$, $\Q \in \mathcal{M}$ where $\Q = \P|_{\Ft{t}}$ and $w \in \LdqF{t}$. Then, \[\EQt{G_{s}(w_t^{s}(\Q,w))}{t} = G_t(w).\]
\end{corollary}
\begin{proof}
This is a special case of lemma~\ref{lemma_exp_penalty} obtained by setting $M = \R^d$ and $A_s = \LdpF{+}$.
\qed\end{proof}

We conclude our discussion on how dual variables across time are related by considering the conditional expectations of the $\alpha_s^{\min}$ and $\Gamma_s$ functions used in the dual representation of conditionally convex risk measures (see corollary~\ref{cor_conditional_dual}).
\begin{lemma}
\label{lemma_exp_cond_penalty}
For any times $0 \leq t < s \leq T$ and if $R_t$ is a closed conditionally convex risk measure, then for any $(\Q,w) \in \olW_t$ with $\Q \in \mathcal{M}^e$, it follows that
\[\cl \EQt{-\alpha_s^{\min}(\Q^s,w_t^s(\Q,w))}{t} = \cl\bigcup_{Z \in A_s} \lrparen{\EQt{Z}{t} + \Gamma_t(w)} \cap M_t.\]
\end{lemma}
\begin{proof}
\begin{proofenum}
\item[$"\subseteq"$]
    \begin{align*}
    &\EQt{-\alpha_s^{\min}(\Q^s,w_t^s(\Q,w))}{t}= \lrcurly{\EQt{u_s}{t}: u_s \in M_s, \trans{w_t^s(\Q,w)}u_s \geq \essinf_{Z \in A_s} \trans{w_t^s(\Q,w)}\EQt{Z}{s} \Pas}\\
    &\quad\quad \subseteq \lrcurly{\EQt{u_s}{t}: u_s \in M_s, \Et{\trans{w_t^s(\Q,w)}u_s}{t} \geq \Et{\essinf_{Z \in A_s} \trans{w_t^s(\Q,w)}\EQt{Z}{s}}{t} \Pas}\\
    &\quad\quad = \lrcurly{u_t \in M_t: \trans{w}u_t \geq \essinf_{Z \in A_s} \trans{w}\EQt{Z}{t} \Pas} = \cl\bigcup_{Z \in A_s} \lrparen{\EQt{Z}{t} + \Gamma_t(w)} \cap M_t.
    \end{align*}
    And since $\cl\bigcup_{Z \in A_s} \inlrparen{\inEQt{Z}{t} + \Gamma_t(w)} \cap M_t$ is closed, this direction is shown.
\item[$"\supseteq"$] Consider a point $u \in \cl\bigcup_{Z \in A_s} \inlrparen{\inEQt{Z}{t} + \Gamma_t(w)} \cap M_t$ and, further, assume $u \not\in
\cl\inEQt{-\alpha_s^{\min}(\Q^s,w_t^s(\Q,w))}{t}$.  Since $\cl \inEQt{-\alpha_s^{\min}(\Q^s,w_t^s(\Q,w))}{t}$ is closed and convex, we
can separate $\{u\}$ and $\cl\inEQt{-\alpha_s^{\min}(\Q^s,w_t^s(\Q,w))}{t}$ by some $v \in \LdqF{t}$, i.e. let $v \in \LdqF{t}$ such that
    \begin{align*}
    \E{\trans{v}u} &< \inf_{z_t \in \cl\EQt{-\alpha_s^{\min}(\Q^s,w_t^s(\Q,w))}{t}} \E{\trans{v}z_t}
    = \inf_{z_s \in -\alpha_s^{\min}(\Q^s,w_t^s(\Q,w))} \E{\trans{w_t^s(\Q,v)}z_s}\\
    &= \E{\essinf_{z_s \in -\alpha_s^{\min}(\Q^s,w_t^s(\Q,w))} \trans{w_t^s(\Q,v)}z_s}.
    \end{align*}
    Note that in the last equality above we can interchange the expectation and infimum since $-\alpha_s^{\min}(\Q^s,w_t^s(\Q,w))$ is decomposable.
    By construction
    \begin{align*}
    \essinf_{z_s \in -\alpha_s^{\min}(\Q^s,w_t^s(\Q,w))} \trans{w_t^s(\Q,v)}z_s = \begin{cases} \essinf_{Z \in A_s} \trans{w_t^s(\Q,v)}\EQt{Z}{s} &\text{on } D\\ -\infty &\text{on } D^c \end{cases}
    \end{align*}
    where $D = \inlrcurly{\omega \in \Omega: G_0(w_t^s(\Q,v)[\omega]) = G_0(w_t^s(\Q,w)[\omega])}$.
    Since $\Q \in \mathcal{M}^e$, one has that $G_0(w_t^s(\Q,v)[\omega]) = G_0(w_t^s(\Q,w)[\omega])$ if and only if $v(\omega) = \lambda(\omega)
w(\omega)$ for some $\lambda \in L^0_t(\R_{++})$ (such that $\lambda w \in \LdqF{t}$).  Thus, it holds $\inE{\essinf_{z_s \in -\alpha_s^{\min}(\Q^s,w_t^s(\Q,w))} \trans{w_t^s(\Q,v)}z_s} > -\infty$ if and only if
    \begin{align*}
    \E{\essinf_{z_s \in -\alpha_s^{\min}(\Q^s,w_t^s(\Q,w))} \trans{w_t^s(\Q,v)}z_s} &= \E{\lambda \essinf_{z_s \in -\alpha_s^{\min}(\Q^s,w_t^s(\Q,w))} \trans{w_t^s(\Q,w)}z_s}= \E{\lambda \essinf_{Z \in A_s} \trans{w}\EQt{Z}{t}}.
    \end{align*}
    But this implies $\inE{\lambda \trans{w}u} < \inE{\lambda \essinf_{Z \in A_s} \trans{w}\inEQt{Z}{t}}$, which is a contradiction to
$u \in \cl\bigcup_{Z \in A_s} \inlrparen{\inEQt{Z}{t} + \Gamma_t(w)} \cap M_t$.
\end{proofenum}
\qed\end{proof}

\begin{corollary}
\label{prop_Gamma}
Let $0 \leq t < s \leq T$, $(\Q,w) \in \olW_t$ with $\Q \in \mathcal{M}^e$. It follows that $\cl\EQt{\Gamma_{s}(w_t^{s}(\Q,w))}{t} = \Gamma_t(w)$.
\end{corollary}
\begin{proof}
This is a special case of lemma~\ref{lemma_exp_cond_penalty} obtained by setting $M = \R^d$ and $A_s = \LdpF{+}$.
\qed\end{proof}

%%%%%%%%%%%%%%%%%%%%%%%%%%%%%%%%%%%%%%%%%%
\section{On the sum of closed acceptance sets and convex upper continuity}
\label{sec_sum_acceptance}
When considering multiportfolio time consistency for closed risk measures we need to guarantee that the composed risk measures are
closed, or else the recursive form would fail to hold.  In particular, this would be true if the sum of acceptance sets are themselves
closed. We will demonstrate the closedness of the sum of convex acceptance sets when the associated dynamic risk measure is convex upper continuous.

Recall that a function $F: X \to \mathcal{P}(Y;C)$ is convex upper continuous (c.u.c.) if $F^{-1}(D) := \inlrcurly{x \in X: F(x) \cap D \neq \emptyset}$
is closed for any closed set $D \in \mathcal{G}(Y;-C)$.

\begin{proposition}
\label{prop_comp_usc}
Let $F: X \to \mathcal{P}(Y;C_Y)$ and $G: Y \to \mathcal{P}(Z;C_Z)$.  If $F,G$ are c.u.c. and $G$ is convex and $-C_Y$-monotone,
then $H: X \to \mathcal{P}(Z;C_Z)$ defined by the composition $H(x) := \bigcup_{y \in F(x)} G(y)$ for any $x
\in X$ is c.u.c.
\end{proposition}
\begin{proof}
For any $D \in 2^Z$, then
\begin{align*}
H^{-1}(D) &= \lrcurly{x \in X: H(x) \cap D \neq \emptyset} = \lrcurly{x \in X: \bigcup_{y \in F(x)} G(y) \cap D \neq \emptyset}\\
&= \lrcurly{x \in X: \exists y \in F(x): G(y) \cap D \neq \emptyset} = \lrcurly{x \in X: F(x) \cap G^{-1}(D) \neq \emptyset}= F^{-1}(G^{-1}(D)).
\end{align*}
Additionally, if $D \in \mathcal{G}(Z;-C_Z)$ then $G^{-1}(D)$ is closed, if $x,y \in G^{-1}(D)$ and $\lambda \in [0,1]$ then $G(\lambda x + (1-\lambda)y) \cap D \neq \emptyset$, and if $x,y \in Y$ such that $x - y \in C_Y$ with $x \in
G^{-1}(D)$ then $y \in G^{-1}(D)$.  This implies that $G^{-1}(D) \in \mathcal{G}(Y,-C_Y)$, and thus $F^{-1}(G^{-1}(D))$ is closed for
any $D \in \mathcal{G}(Z;-C_Z)$.
\qed\end{proof}

\begin{lemma}
\label{lemma_closed_sum_acceptance}
Let $M_t$ ($M_{s}$) be the set of eligible portfolios at time $t$ ($s$) (a closed linear subspace of $\LdpF{t}$ ($\LdpF{s}$)).  Let $R_{t,s}$ be a c.u.c. convex stepped risk measure from $t$ to $s$ and $R_{s}$ be a c.u.c. risk measure at time $s$.  Then, $A_{t,s} + A_{s}$ is closed.
\end{lemma}
\begin{proof}
By lemma 3.6(i)~in~\cite{FR12}, $A_{t,s} + A_{s} = \inlrcurly{X \in \LdpF{}: 0 \in \bigcup_{Z \in R_{s}(X)}
R_{t,s}(-Z)}$.  Indeed, \begin{align*}
X\in A_{t,s} + A_{s} &\Leftrightarrow -R_{s}(X)\cap A_t\neq\emptyset\Leftrightarrow\exists Z\in R_{s}(X) \mbox{ s.t. } -Z\in A_t ~(\mbox{ i.e. } 0\in R_t(-Z) =R_{t,s}(-Z))\\
&\Leftrightarrow 0 \in \bigcup_{Z \in R_{s}(X)}R_{t,s}(-Z).
\end{align*}
Let $\tilde{R}_t(X) := \bigcup_{Z \in R_{s}(X)} R_{t,s}(-Z)$ then $A_{t,s} + A_{s} =
\tilde{R}_t^{-1}(M_{t,-})$. By proposition~\ref{prop_comp_usc}, $\tilde{R}_t$ is c.u.c., and thus $\tilde{R}_t^{-1}(M_{t,-})$ is closed.
\qed\end{proof}

\begin{remark}
\label{rem_stepped_usc}
Let $R_t$ be a conditional risk measure at time $t$ and $R_{t,s} := R_t|_{M_{s}}$ be the stepped risk measure from $t$ to $s$ associated with $R_t$.  If $R_t$ is c.u.c. then, trivially, $R_{t,s}$ is c.u.c.
\end{remark}

Moreover, when applying lemma~\ref{separation} to the proof of theorem~\ref{thm_mptc_penalty_wo_uc} and corollary~\ref{thm_mptc_penalty} we need not only the sum of closed convex acceptance sets to be closed, but also to be a (closed) convex acceptance set itself.  This is given in the following lemma.

\begin{lemma}
\label{lemma_sum_acceptance}
Let $\seq{A}$ be a sequence of closed convex normalized acceptance sets.
Assume $A_{t,t+1} + A_{t+1} \subseteq A_t$, then $A_{t,t+1} + A_{t+1}$ is a convex acceptance set at time $t$. Furthermore, if $\seq{A}$ is c.u.c., then $A_{t,t+1} + A_{t+1}$ is closed.
\end{lemma}
\begin{proof} Let us check the properties of acceptance sets (see definition~\ref{defn_acceptance}).
\begin{proofenum}
\item $A_{t,t+1} + A_{t+1} \subseteq \LdpF{}$ trivially.
\item $M_t \cap \inlrparen{A_{t,t+1} + A_{t+1}} \supseteq M_t \cap M_{t+1} \cap A_t \neq \emptyset$ since $0 \in A_{t+1}$ (by $A_{t+1}$ closed and normalized), $M_t \cap A_t \neq \emptyset$, and $M_t \cap M_{t+1} = M_t$.
\item $M_t \cap \inlrparen{\LdpF{} \backslash \inlrcurly{A_{t,t+1} + A_{t+1}}} \supseteq M_t \cap \inlrparen{\LdpF{} \backslash A_t} \neq \emptyset$ by $A_{t,t+1} + A_{t+1} \subseteq A_t$.
\item $A_{t,t+1} + A_{t+1} + \LdpF{+} \subseteq A_{t,t+1} + A_{t+1}$ trivially.
\end{proofenum}
$A_{t,t+1} + A_{t+1}$ is convex since both $A_{t,t+1}$ and $A_{t+1}$ are convex. $A_{t,t+1} + A_{t+1}$ is closed by lemma~\ref{lemma_closed_sum_acceptance} if $\seq{A}$ is c.u.c.
\qed\end{proof}

We will finish this section by considering a class of risk measures which are point plus cone and show that these risk measures will be c.u.c. under $p = +\infty$ and the weak* topology.
\begin{proposition}
\label{prop_uc_scalar}
Consider the full eligible space $M_t = \LdiF{t}$ and let $p = +\infty$.  Let $R_t(X) := \rho_t(X) + \LdiF{t,+}$ for some vector $\rho_t$ of
scalar conditional risk measures, i.e. $\rho_t(X) := \transp{(\rho_t)_1(X_1),...,(\rho_t)_d(X_d)}$.  If $\rho_t$ is (componentwise) lower semicontinuous and convex then $R_t$ is c.u.c.
\end{proposition}
\begin{proof}
Recall from the scalar literature that $\rho_t(X) \in \LdiF{t}$ for any $X \in \LdiF{}$.  Consider any set $D \in \mathcal{G}(\LdiF{t};\LdiF{t,-})$.  It follows that
\begin{align*}
R_t^{-1}(D) &= \lrcurly{X \in \LdiF{}: R_t(X) \cap D \neq \emptyset}= \lrcurly{X \in \LdiF{}: \exists \hat{d} \in D, \rho_t(X) \preceq \hat{d}}\\
&= \lrcurly{X \in \LdiF{}: \exists \hat{d} \in D, \rho_t(X) = \hat{d} \Pas}= \lrcurly{X \in \LdiF{}: \rho_t(X) \in D} = \rho_t^{-1}(D).
\end{align*}
Therefore we wish to show that $\rho_t^{-1}(D)$ is weak* closed.  From $\rho_t$ convex, it immediately follows that $\rho_t^{-1}(D)$
is convex, therefore $\rho_t^{-1}(D)$ is weak* closed if and only if $\rho_t^{-1}(D) \cap \inlrcurly{Z \in \LdiF{}: \|Z\|_{\infty} \leq
k}$ is closed in probability for every $k$ by~\cite[proposition 5.5.1]{KS09}.  Pick any $k \geq 0$ and consider
\[(Z_n)_{n \in \N} \subseteq \rho_t^{-1}(D) \cap \lrcurly{Z \in \LdiF{}: \|Z\|_{\infty} \leq k}\]
with $Z_n \to \bar Z$ in probability (and
thus $\bar Z \in \inlrcurly{Z \in \LdiF{}: \|Z\|_{\infty} \leq k}$).
Note that convergence in probability implies there exists a subsequence
which converges almost surely, we will denote this subsequence by $(Z_{n_m})_{m \in \N} \to \bar Z$.  For any sequence of random
vectors $(Y_n) \subseteq \LdiF{}$, define $\liminf_{n \to \infty} \rho_t(Y_n) = \lim_{n \to \infty} \inf_{m \geq n} \rho_t(Y_m)$ where \[\inf_{m \geq n} \rho_t(Y_m) = \lrparen{\begin{array}{c} \inf_{m \geq n} (\rho_t)_1((Y_m)_1) \\ \vdots \\ \inf_{m \geq n} (\rho_t)_d((Y_m)_d) \end{array}}.\]
Since $D$ is a lower set and $\inf_{\hat{m} \geq m} \rho_t(Z_{n_{\hat{m}}}) \preceq \rho_t(Z_{n_m})$ (and $\rho_t(Z_{n_m}) \in D$) for
any $m \in
\N$, then it follows that $\inf_{\hat{m} \geq m} \rho_t(Z_{n_{\hat{m}}}) \in D$ for any $m \in \N$. Note that $\|\inf_{\hat{m}
\geq m} \rho_t(Z_{n_{\hat{m}}})\|_{\infty}
\leq \max(\|\rho_t(0)+k\|_{\infty},\|\rho_t(0)-k\|_{\infty}) =: \hat{k}$ by $\|Z_{n_{\hat{m}}}\|_{\infty} \leq k$ for every $\hat{m}
\in \N$.  Since $D \cap \inlrcurly{u \in \LdiF{t}:
\|u\|_{\infty} \leq \hat{k}}$ is closed in probability (by~\cite[proposition 5.5.1]{KS09}) it must contain all almost sure limit
points, therefore we have that $\liminf_{m \to \infty}
\rho_t(Z_{n_m}) \in D \cap \inlrcurly{u \in \LdiF{t}: \|u\|_{\infty} \leq \hat{k}}$.  Finally from componentwise lower semicontinuity we
have $\liminf_{m \to \infty} \rho_t(Z_{n_m}) \succeq \rho_t(Z)$, therefore by $D$ a lower set it follows that $\rho_t(\bar Z) \in
D$, i.e. $\bar Z \in \rho_t^{-1}(D)$.
\qed\end{proof}

%%%%%%%%%%%%%%%%%%%%%%%%%%%%%%%%%%%%
\section{Stepped risk measures}
\label{sec_stepped}
In this section, we consider the dual representation of closed convex and coherent stepped risk measures $R_{t,s}: M_{s} \to \mathcal{P}(M_t;M_{t,+})$. This is used in sections~\ref{sec_convex} and \ref{sec_coherent} as the stepped penalty functions and stepped sets of dual variables play a role when discussing equivalent characterizations of multiportfolio time consistency.  For the dual representation we will use set-valued duality defined in~\cite{H09} analogously as for conditional risk measures in section 4 of~\cite{FR12}.

Given a risk measure $R_t: \LdpF{} \to \mathcal{P}(M_t;M_{t,+})$,
a stepped risk measure is the restriction of $R_t$ to $M_{s}$, i.e. $R_{t,s}=R_t|_{M_{s}}$.
The primal representation can immediately be seen, that is $R_{t,s}(X) := \{u \in M_t: X + u \in A_{t,s}\}$ for $X \in M_{s}$.
Therefore, if $R_t$ is closed convex (coherent) then $R_{t,s}$ is closed convex (coherent). Furthermore, if $R_t$ is $\LdpF{+}$-monotone, then $R_{t,s}$ is $M_{s,+}$-monotone.

\begin{lemma}
\label{lemma_stepped_dual}
Let $R_t$ be a closed convex risk measure.
The set of dual variables for
$R_{t,s}: M_{s} \to \mathcal{P}(M_t;M_{t,+})$ with $t < s$ is given by
\[\olW_{t,s} = \lrcurly{(\Q,w) \in \mathcal{M} \times \lrparen{\plus{M_{t,+}} \backslash M_t^{\perp}}: w_t^{s}(\Q,w) \in \plus{M_{s,+}}, \Q = \P|_{\Ft{t}}}\]
\end{lemma}
\begin{proof}
By logic of proposition 4.4 in~\cite{FR12} the set of (classical) stepped dual variables are given
by $\inlrcurly{(Y,v): Y \in \plus{M_{s,+}}, v \in \inlrparen{\inEt{Y}{t} + M_t^{\perp}} \backslash
M_t^{\perp}}$.  Then it remains to show that for any dual pair $(Y,v)$ there exists a
$(\Q,w) \in \olW_{t,s}$ such that $\FYv{X} = \FQw{X}$ for any $X \in M_{s}$, and vice versa,
where $\FYv{X} := \inlrcurly{u \in M_t: \inE{\trans{X}Y} \leq \inE{\trans{v}u}}$.
\begin{proofenum}
    \item Let $(\Q,w) \in \olW_{t,s}$. Then, we will show that there exists a dual pair \[(Y,v) \in \lrcurly{(Y,v): Y \in
\plus{M_{s,+}}, v \in \lrparen{\Et{Y}{t} + M_t^{\perp}} \backslash M_t^{\perp}}\] such that $\FYv{X} = \FQw{X}$ for any $X \in M_{s}$.
Let $Y = w_t^{s}(\Q,w) \in \plus{M_{s,+}}$ (by remark~\ref{rem_stepped_dual_var} and lemma~\ref{lemma_Wtau_to_Wt}~(i)), thus
$\inE{\trans{X}Y} = \inE{\trans{w_t^{s}(\Q,w)}X} = \inE{\trans{w}\inEQt{X}{t}}$ and $\inEt{Y}{t}=w$.  From $w \in \plus{M_{t,+}}
\backslash M_t^{\perp}$ we can rewrite $w = w_{\plus{M_{t,+}}} + w_{M_t^{\perp}}$.  Thus $v = w_{\plus{M_{t,+}}} = w - w_{M_t^{\perp}}
\in \inEt{Y}{t} + M_t^{\perp}$.  Finally, $w \not\in M_t^{\perp}$ implies $v \not\in M_t^{\perp}$, and $\inE{\trans{w}u} = \inE{\trans{v}u}$ for every $u \in M_t$ since $w \in v + M_t^{\perp}$.

    \item Let $Y \in \plus{M_{s,+}}$ and $v \in \inlrparen{\inEt{Y}{t} + M_t^{\perp}} \backslash M_t^{\perp}$.
    We want to show there exists a $(\Q^t,w) \in \olW_{t,s}$ such that $\FYv{X} = \FQw{X}$ for any $X \in M_{s}$.  First we will
let $w \in \inEt{(Y + M_{s}^{\perp}) \cap \LdqF{s,+}}{t} $
    (which is nonempty), i.e., $w = \inEt{Y + m^{\perp}}{t}$ for some $m^{\perp} \in M_{s}^{\perp}$ and  $Y + m^{\perp} \in
\LdqF{s,+}$.  Then it can easily be seen that $w \in v + M_t^{\perp}$ for $v \in \inlrparen{\inEt{Y}{t} + M_t^{\perp}} \backslash
M_t^{\perp} \subseteq \plus{M_{t,+}}$.  Thus $w \in \plus{M_{t,+}} + M_t^{\perp}$ and with $v \not\in M_t^{\perp}$ this implies $w \in
\plus{M_{t,+}} \backslash M_t^{\perp}$.  From $w \in v + M_t^{\perp}$ it follows that $\inE{\trans{w}u} = \inE{\trans{v}u}$ for every $u \in M_t$.

        Additionally, choose $\Q\in\mathcal{M} $ such that $\dQidP = \bar{\xi}_{0,s}(\Q_i)$ where
        \[
        \bar{\xi}_{r,s}(\Q_i)[\omega] = \begin{cases}\frac{\Et{Y_i + m_i^{\perp}}{s}(\omega)}{\Et{Y_i + m_i^{\perp}}{r}(\omega)} &\text{if } \Et{Y_i + m_i^{\perp}}{r}(\omega) > 0\\ 1 &\text{else}\end{cases}
        \] for any $0 \leq r \leq s$ and almost every $\omega \in \Omega$. Define the measure $\Q^t\in \mathcal{M}$  by its density $\frac{d\Q^t_i}{d\P} = \bar{\xi}_{t,s}(\Q_i)$.
        Then $w_t^{s}(\Q^t,w) = w_t^{s}(\Q,w) = Y + m^{\perp} \in \plus{M_{s,+}} + M_{s}^{\perp} \subseteq \plus{M_{s,+}}$.  Therefore,
$\inE{\trans{w}\inEQt{X}{t}} = \inE{\trans{w_t^{s}(\Q^t,w)}X} = \inE{\trans{Y}X}$ for every $X \in M_{s}$.
\end{proofenum}
\qed\end{proof}

\begin{remark}
\label{rem_stepped_dual_var}
For any choice of eligible portfolios $M_t$, it follows that $\olW_{t,s} \supseteq \olW_t$ for any $t < s$.
\end{remark}

\begin{remark}
\label{rem_stepped_dual_var_equal}
If we consider the case when $M_t = \LdpF{t}$ for all times $t$, then an inspection of the proof of lemma 4.5 from~\cite{FR12} shows that $\olW_{t,s} = \olW_t$.
\end{remark}

The lemma below gives a dual representation for closed convex stepped risk measures.  In particular, it demonstrates that the minimal stepped penalty function as defined in \eqref{stepped_beta} can be used in a  dual representation  to define a closed convex stepped risk measures.

\begin{lemma}
\label{lemma_stepped_rep}
The dual representation for any closed convex stepped risk measure $R_{t,s}: M_{s} \to \mathcal{G}(M_t;M_{t,+})$ with $t < s$ is given by
\[R_{t,s}(X) = \bigcap_{(\Q,w) \in \olW_{t,s}} \lrsquare{-\beta_{t,s}^{\min}(\Q,w) + \lrparen{\EQt{-X}{t} + G_t(w)} \cap M_t}\]
for any $X \in M_{s}$ where
\[-\beta_{t,s}^{\min}(\Q,w) = \cl \bigcup_{X \in A_{t,s}} \lrparen{\EQt{X}{t} + G_t(w)} \cap M_t.\]
\end{lemma}
\begin{proof}
This is an adaption of theorem~\ref{thm_probability_equal} to stepped risk measures using lemma~\ref{lemma_stepped_dual}.
\qed\end{proof}

We will use the above results to give a dual representation for closed coherent stepped risk measures.

\begin{corollary}
\label{cor_stepped_rep}
The dual representation for any closed coherent stepped risk measure $R_{t,s}: M_{s} \to \mathcal{G}(M_t;M_{t,+})$ with $t < s$ is given by
\[R_{t,s}(X) = \bigcap_{(\Q,w) \in \olW_{t,s}^{\max}} \lrparen{\EQt{-X}{t} + G_t(w)} \cap M_t\]
for any $X \in M_{s}$ where
\[\olW_{t,s}^{\max} = \lrcurly{(\Q,w) \in \olW_{t,s}: w_t^{s}(\Q,w) \in \plus{A_{t,s}}}.\]
\end{corollary}
\begin{proof}
Note that $-\beta_{t,s}^{\min}(\Q,w) = \cl \bigcup_{X \in A_{t,s}} \inlrparen{\inEQt{X}{t} + G_t(w)} \cap M_t = G_t(w) \cap M_t$ if and only if for every $X \in A_{t,s}$ we have
\[\E{\trans{w}\EQt{X}{t}} = \E{\trans{w_t^{s}(\Q,w)}X} \geq 0,\]
i.e. $w_t^{s}(\Q,w) \in \plus{A_{t,s}}$.
Thus, for a $M_{s,+}$-monotone closed coherent stepped risk measure $R_{t,s}$ with  $0 \leq t < s \leq T$ it holds that for any $(\Q,w) \in \olW_{t,s}$
\[-\beta_{t,s}^{\min}(\Q,w) = G_t(w) \cap M_t \Leftrightarrow w_t^{s}(\Q,w) \in \plus{A_{t,s}}.\]
An application of lemma~\ref{lemma_stepped_dual} provides the desired result.
\qed\end{proof}

Finally, we will use the above duality results to extend corollary~\ref{cor_conditional_dual} to the stepped risk measures.
\begin{corollary}
\label{cor_conditional_stepped_rep}
The dual representation for any closed conditionally convex stepped risk measure $R_{t,s}: M_{s} \to \mathcal{G}(M_t;M_{t,+})$ with $t < s$ is given by
\[R_{t,s}(X) = \bigcap_{(\Q,w) \in \olW_{t,s}} \lrsquare{-\alpha_{t,s}^{\min}(\Q,w) + \lrparen{\EQt{-X}{t} + \Gamma_t(w)} \cap M_t}\]
for any $X \in M_{s}$ where
\[-\alpha_{t,s}^{\min}(\Q,w) = \cl \bigcup_{X \in A_{t,s}} \lrparen{\EQt{X}{t} + \Gamma_t(w)} \cap M_t.\]
If $R_{t,s}$ is additionally conditionally coherent then
\[R_{t,s}(X) = \bigcap_{(\Q,w) \in \olW_{t,s}^{\max}} \lrparen{\EQt{-X}{t} + \Gamma_t(w)} \cap M_t.\]
\end{corollary}
\begin{proof}
This is an adaption of corollary~\ref{cor_conditional_dual} to stepped risk measures using the results of lemma~\ref{lemma_stepped_rep} and corollary~\ref{cor_stepped_rep}.
\qed\end{proof}

\bibliographystyle{spmpsci}
\bibliography{biblio-FS}
\end{document}